\documentclass[10pt,final]{IEEEtran}

\usepackage{float}
\usepackage{algorithm}
\usepackage[noend]{algorithmic}
\usepackage[geometry]{ifsym}
\usepackage{bbm}
\usepackage{amssymb}
\usepackage{yhmath}
\usepackage{amscd}
\usepackage{dsfont}
\usepackage{amsmath}
\usepackage{graphicx}
\usepackage{psfrag}
\usepackage{rotating}
\usepackage{amsthm}
\usepackage{amsfonts}
\usepackage{url}
\usepackage{color}
\usepackage{balance} 
\usepackage{mathtools}
\usepackage{tabularx}
\usepackage{booktabs}
\usepackage{enumitem} 
\usepackage[labelfont=bf,format=plain,justification=raggedright,singlelinecheck=false]{caption}

\newtheorem{lemma}{Lemma}

\newtheorem{theorem}{Theorem}

\newtheorem{proposition}{Proposition}

\newtheorem{corollary}{Corollary}
\newtheorem{remark}{Remark}

\newcommand{\blue}[1]{#1}

\usepackage{bm}
\usepackage{placeins}
\providecommand{\Hel}{\mathrm{H}}
\providecommand{\Tr}{\operatorname{Tr}}
\providecommand{\DF}{\mathrm{DF}}
\newenvironment{bl}{}{}
\newcommand{\vphi}{\varphi}

\newcommand{\Kb}{\bar K}
\newcommand{\Ct}{\tilde C_2}
\newcommand{\Th}{\Theta_{\mathrm{eff}}}

\begin{document}
\title{Time-Switching Energy-Harvesting Quantum Relay-Assisted FSO Communication with Displaced Squeezed Signaling\thanks{A preliminary version of this work was presented at IEEE QCNC 2026 \cite{KhalfetQCNC2026}. This manuscript substantially extends that work.}}

 \author{Nizar Khalfet, \IEEEmembership{Member,~IEEE,} Francisco A. Monteiro, \IEEEmembership{Member,~IEEE,} and Ioannis Krikidis, \IEEEmembership{Fellow,~IEEE}
\thanks{N. Khalfet and I. Krikidis are with Department of Electrical and Computer Engineering, University of Cyprus, Cyprus (e-mail: \{khalfet.nizar,~krikidis\}@ucy.ac.cy).}
\thanks{F.~A.~Monteiro~is~with~Instituto~de~Telecommunica\c{c}\~{o}es,~and~\mbox{ISCTE-Instituto} Universit\'{a}rio de Lisboa, Portugal (e-mail: francisco.monteiro@lx.it.pt). }

}
\maketitle


\begin{abstract}
\begin{bl}
This paper investigates a quantum relay-assisted free-space optical (FSO) communication system in which a decode-and-forward relay is powered by time-switching optical energy harvesting and forwards the decoded information using displaced squeezed signaling (DSS). To obtain a physically consistent description of turbulence-induced loss for non-classical signaling, each hop is modeled as a pure-loss bosonic channel whose random transmissivity follows a truncated Gamma-Gamma law, i.e., the Gamma-Gamma irradiance conditioned on the physically admissible range of a transmittance. Under this channel, coherent states remain pure, whereas squeezed states are mapped to mixed Gaussian states. Under Helstrom-optimal detection, the average error of the coherent hop is derived in exact Laplace-series form. For the DSS hop, we derive the output states, a closed-form Gaussian fidelity with rigorous two-sided bounds on the Helstrom error, and prove the existence of a loss-aware optimal photon partition. Furthermore, we prove that optimized squeezing preserves both the turbulence-induced diversity order and the leading asymptotic constant of coherent relaying, so that its benefit is a finite-photon-budget effect, and that an interior throughput-optimal harvesting fraction exists. A covariance-preserving  benchmark with an exact series representation yields an analytical upper bound on the achievable throughput whose absolute gap vanishes as the relay photon budget grows. Numerical results based on the exact mixed-state Helstrom error quantify the benefit of squeezing and the regime in which it is useful.
\end{bl}
\end{abstract}


\section{Introduction}

Quantum communication exploits the fundamental principles of quantum
mechanics to enable information transmission beyond the limits imposed
by classical communication systems \blue{\cite{b1,Hanzo2025}}. In contrast to classical channels,
quantum channels allow the transmission of both classical and quantum
information while exploiting uniquely quantum resources such as \blue{entanglement and}
non-orthogonal quantum states \blue{\cite{b1,Hanzo2025,SchumacherWestmoreland}}.
These properties fundamentally modify the information-theoretic
limits of communication systems and introduce new performance
metrics governed by the laws of quantum detection and estimation
theory \blue{\cite{HelstromBook,b14,Pirandola2017}}.

Within continuous-variable (CV) \cite{Braunstein2005}  quantum communication systems,
optical coherent states constitute the most widely adopted signaling
format due to their practical compatibility with existing photonic
technologies and coherent detection receivers
\blue{\cite{Weedbrook2012}}. However, coherent-state
discrimination is fundamentally limited by quantum noise,
particularly in photon-limited regimes where the overlap between
the quantum signal states becomes significant.

To mitigate this limitation, recent research has investigated
non-classical optical states that exploit additional quantum
resources beyond coherent displacement. Among these, displaced
squeezed states (DSS) have emerged as a promising signaling format
that combines coherent displacement with quadrature squeezing
\cite{Chesi2018}. By redistributing quantum noise between
orthogonal quadratures while preserving the Heisenberg uncertainty
relation, squeezing modifies the geometry of the quantum signal
constellation and improves state distinguishability under optimal
quantum detection \blue{\cite{Weedbrook2012,Chesi2018,Paris2001}}.

Free-space optical (FSO) communication constitutes a key physical
platform for long-distance quantum communication networks
\cite{FSO}. Optical free-space links enable high-bandwidth and
license-free transmission and are widely considered a promising
technology for ground-to-satellite and satellite-to-ground quantum
communication systems. Nevertheless, atmospheric turbulence \cite{AlHabash2001} introduces
random irradiance fluctuations and phase distortions that significantly
degrade the reliability of quantum state discrimination. \blue{Atmospheric turbulence-induced irradiance fluctuations are commonly modeled by the Gamma-Gamma distribution under appropriate scintillation regimes, which accounts for both small-scale and large-scale scintillation effects \cite{AlHabash2001,Andrews2005,BhatnagarRelay}.} Even under optimal quantum measurements \cite{HelstromBook}
turbulence-induced fading fundamentally limits the achievable
reliability of long-distance quantum optical links.

Relay-assisted architectures have been proposed as an effective
mechanism for mitigating turbulence impairments by partitioning long
propagation paths into multiple shorter hops. In this context,
relay-assisted quantum FSO communication has recently been investigated
in \cite{BhatnagarRelay}, where a rigorous analytical framework for
decode-and-forward (DF) quantum relays operating over
Gamma-Gamma turbulence channels was established. The analysis derived
closed-form Helstrom error expressions for coherent-state signaling and
demonstrated the reliability benefits achievable through relay
deployment. However, existing relay-assisted quantum communication
studies typically assume externally powered relay nodes and therefore
do not account for the energy sustainability constraints that arise in
practical autonomous quantum networks.

Energy harvesting (EH) has recently emerged as an attractive solution
for enabling self-sustaining wireless communication systems
\cite{Ulukus2019, Varshney2008, Nasir2013}. In classical wireless networks,
simultaneous wireless information and power transfer (SWIPT)
architectures have been extensively studied as a mechanism for
supporting energy-constrained communication nodes. More recently, similar ideas have begun to appear in quantum
communication systems. In particular, wireless-powered quantum
communication architectures have been investigated in
\cite{KrikidisTesla}, where the integration of EH
and quantum state discrimination was first explored.

While the above studies provide important insights into different
aspects of quantum optical communication, a unified framework that
jointly captures relay-assisted transmission, EH
constraints, and non-classical optical signaling is still missing.
In particular, the relay-assisted quantum FSO framework developed
in~\cite{BhatnagarRelay} focuses on externally powered relays and
coherent-state signaling, while wireless-powered quantum optical
architectures such as~\cite{KrikidisTesla} consider
EH nodes but do not incorporate relay-assisted
turbulence-impaired links. On the other hand, prior works on
squeezing-enhanced quantum communication~\cite{Chesi2018}
primarily investigate point-to-point systems and do not address
energy-constrained relay operation under atmospheric turbulence.
Consequently, the interplay between turbulence-induced fading,
energy causality at relay nodes, and non-classical signaling
remains largely unexplored. This work bridges this gap by developing a unified analytical framework for relay-assisted quantum FSO communication with optical EH and DSS signaling.
 The proposed model enables a fundamental characterization of
the reliability and achievable throughput of energy-constrained
quantum relay systems operating over turbulence-impaired optical
channels under Helstrom-optimal quantum detection.

\begin{bl}
Squeezing-enhanced binary communication has been studied for point-to-point links in the lossless setting, where the optimal split of a fixed photon budget between displacement and squeezing is known in closed form \cite{Chesi2018, Paris2001,Walsh2025}. In atmospheric channels, however, turbulence acts on the optical field as a random loss, which is described by a probability distribution of transmittance (PDT) \cite{SemenovVogel2009,Vasylyev2016}. Since a vacuum mode necessarily enters through every loss port, a squeezed state is mixed by any nontrivial transmissivity, even when the transmissivity is perfectly known. The reliability analysis of a DSS relay must therefore be carried out with mixed-state Helstrom detection, which distinguishes the present problem from the lossless one. Following the truncation-and-renormalization principle used for bounded atmospheric PDTs in \cite{Vasylyev2018}, we adopt a tractable truncated Gamma-Gamma model, in which the Gamma-Gamma law of the scaled irradiance is conditioned on the physically admissible range of the transmissivity (at most one) and renormalized. The application of this construction to Gamma-Gamma irradiance is a modeling approximation of the present work.
\end{bl}

\begin{bl}
The discrimination of mixed Gaussian states has been addressed through the quantum Chernoff bound \cite{Audenaert2007} and its computable forms for Gaussian states \cite{Calsamiglia2008,Pirandola2008,Tan2008}, which characterize the error of a given pair of mixed states for fixed channel parameters. The present paper builds on these tools, together with the fidelity and the Fuchs--van de Graaf inequalities, and develops the analysis at the system level, where its contribution lies. First, the mixedness is derived from the physical pure-loss channel of an FSO hop: squeezed relay states are mixed, whereas coherent states remain pure. Second, the exact Helstrom error is averaged over the random transmissivity of the truncated Gamma-Gamma model, which allows us to prove that a loss-aware photon allocation exists and differs from the lossless rule. Third, we establish a fading-specific result: optimized squeezing preserves the diversity order and the leading constant of coherent relaying, and the covariance-preserving (CP) benchmark isolates what the vacuum noise removes. Fourth, these results are embedded in a time-switching energy-harvesting DF relay, for which we characterize the throughput-optimal harvesting design and a rigorous throughput upper bound.
\end{bl}

\begin{bl}
The main contributions of this work are summarized as follows.

\begin{itemize}

\item \textbf{Physically consistent framework for EH quantum relays:}
We develop an analytical framework for time-switching decode-and-forward (TS-DF) quantum relay-assisted FSO communication with optical EH, in which each hop is a pure-loss bosonic channel with truncated Gamma-Gamma-distributed transmissivity and detection is Helstrom-optimal. Coherent inputs remain pure, whereas DSS inputs become mixed Gaussian states.

\item \textbf{Exact characterization of the DSS hop:}
We derive the Gaussian output states of the DSS hop, a closed-form fidelity, and rigorous two-sided bounds on the Helstrom error, with the lower bound becoming exact for pure-state outputs (fidelity tools tailored to the pure-loss channel). For coherent signaling, we obtain an exact Laplace-series representation of the average error over Gamma-Gamma turbulence, with an explicit bound on the effect of the truncation.

\item \textbf{Loss-aware squeezing allocation:}
We prove the existence of the photon partition that minimizes the exact average Helstrom error, and show that it does not reduce to the lossless rule of \cite{Chesi2018, Paris2001,Walsh2025}: under turbulence-induced loss, the lossless rule can yield a larger error than coherent relaying, while the loss-aware allocation is a small squeezing fraction.

\item \textbf{Asymptotic reliability under loss:}
We prove that optimized squeezing preserves the turbulence-induced diversity order and the leading constant of coherent relaying, and that the optimal squeezing fraction vanishes as the photon budget grows. We also show that a  CP benchmark would double the diversity order, which quantifies what is lost through the vacuum noise entering the loss port.

\item \textbf{Throughput and harvesting design:}
We show that an interior throughput-optimal harvesting fraction exists, and that the CP benchmark is a rigorous throughput upper bound whose absolute gap vanishes as the photon budget grows.

\end{itemize}
\end{bl}

\emph{Notation:}
Scalars are denoted by lowercase letters, while random variables are
denoted by uppercase letters, except the random transmissivities $\eta_i$. The realization of a random variable
$X$ is denoted by $x$. Sets are denoted by calligraphic letters, e.g.,
$\mathcal{X}$. The expectation operator with respect to the
distribution of a random variable $X$ is denoted by $\mathbb{E}[\cdot]$.
The notation $\log_2(\cdot)$ denotes the logarithm with base $2$, while
$\ln(\cdot)$ denotes the natural logarithm. The operator
$\mathrm{Tr}(\cdot)$ denotes the trace of a matrix or linear operator. Quantum states are represented by density operators $\rho$ acting on a Hilbert space, while pure quantum states are denoted using Dirac's braket
notation $|\psi\rangle$. The inner product between two quantum states is
denoted by $\langle \psi | \phi \rangle$. The notation $\{\Pi_m\}$
denotes a positive operator-valued measure (POVM) used for quantum
measurement. The operator $\|\cdot\|_1$ denotes the trace norm. The
symbol $\mathbb{P}(\cdot)$ denotes probability.

The remainder of this paper is organized as follows. Section II introduces the system model and the optical EH architecture. Section III formulates the end-to-end detection problem under Helstrom-optimal receivers. Section IV presents the performance analysis and derives the structural properties of the optimal design. Numerical results are provided in Section V, followed by concluding remarks in Section VI.

\section{System Model}
\label{sec:system}

We consider a relay-assisted quantum FSO
communication system composed of three nodes: a source $S$, a relay
$R$, and a destination $D$, as depicted in
Fig.~\ref{fig:systemRelay}. Communication takes place over two
consecutive optical hops, namely $S\!\to\!R$ and $R\!\to\!D$. Both hops
are impaired by atmospheric turbulence. 
The relay operates according to a TS-DF protocol and is assumed to be batteryless. Each transmission block is normalized to unit duration,
so that the time-switching parameter directly represents
the fraction of time allocated to EH.

Under the TS-DF protocol, a fraction $t\in(0,1)$ of each transmission
block is allocated to EH at the relay, remaining fraction $1-t$ is used for information transmission, which, under the half-duplex protocol of Remark~1, is split into an S$\to$R slot and an R$\to$D slot of duration $(1-t)/2$ each. All
energy harvested within a given block is consumed within the same
block. Hence, the relay satisfies an instantaneous energy-causality
constraint and no inter-block energy storage is assumed.
\begin{remark}
The proposed TS-DF relay operates in a half-duplex mode. Within each
transmission block, energy harvesting and information transmission are orthogonalized via time-switching, resulting in two non-overlapping phases. During the communication phase, the S$\rightarrow$R and R$\rightarrow$D links are executed sequentially, ensuring that the relay
does not transmit and receive simultaneously. This assumption is
consistent with batteryless optical relay architectures and avoids the optical isolation and self-interference mitigation requirements
associated with full-duplex operation.
\end{remark}

\begin{figure}[!t]
\centering
\includegraphics[width=\columnwidth,trim={0 2mm 0 2mm},clip]{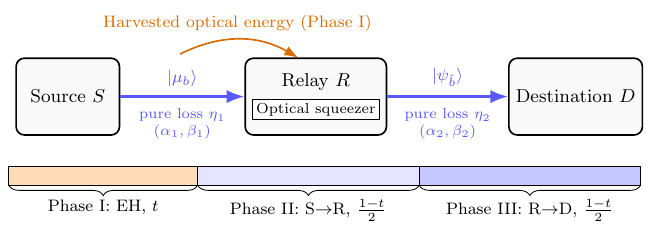}
\caption{System model of the proposed TS-DF quantum relay-assisted FSO link. Each block is split into an EH phase of duration $t$ and two forwarding slots of duration $(1-t)/2$ each (half-duplex, Remark~1); each hop is a pure-loss channel with truncated Gamma-Gamma transmissivity $\eta_i$.}
\label{fig:systemRelay}
\end{figure}

\subsection{Source Signaling}

The source encodes classical information into quantum states of a
single bosonic optical mode. Let $\mathcal{H}$ denote the Hilbert space
of the mode, and let $\hat a$ and $\hat a^\dagger$ denote the
annihilation and creation operators satisfying the canonical
commutation relation
\begin{equation}
[\hat a,\hat a^\dagger]=1.
\label{eq:CCR}
\end{equation}
A coherent state $|\alpha\rangle\in\mathcal{H}$ is defined as the
eigenstate of the annihilation operator, i.e.,
\begin{equation}
\hat a|\alpha\rangle=\alpha|\alpha\rangle,
\label{eq:coherent_eigen}
\end{equation}
where $\alpha\in\mathbb{C}$ denotes the complex field amplitude. The
corresponding mean photon number is $|\alpha|^2$.

Throughout the paper, the source employs 
BPSK coherent-state modulation. The source alphabet is therefore
given by
\begin{equation}
\mathcal{E}_S
=
\left\{
|\mu_0\rangle,|\mu_1\rangle
\right\}
=
\left\{
|+\sqrt{a}\rangle,\,
|-\sqrt{a}\rangle
\right\},
\label{eq:source_alphabet_final}
\end{equation}
where $a>0$ denotes the mean photon number per transmitted symbol.
Unless otherwise stated, the two symbols are assumed equiprobable.

Since coherent states are non-orthogonal, they cannot be perfectly
distinguished by any physically admissible quantum measurement. Their
pairwise overlap is
\begin{equation}
\langle\mu_0|\mu_1\rangle
=
\exp(-2a),
\label{eq:source_overlap_final}
\end{equation}
which determines the fundamental discrimination limit on the first hop.

\subsection{FSO Channel  Model}

For hop $i\in\{1,2\}$, let $I_i$ denote the normalized irradiance
coefficient induced by atmospheric turbulence. The random variable
$I_i$ follows the Gamma-Gamma distribution with parameters
$(\alpha_i,\beta_i)$:
\begin{IEEEeqnarray}{l}
\nonumber
f_{I_i}(x)
=
\frac{2(\alpha_i\beta_i)^{\frac{\alpha_i+\beta_i}{2}}}
{\Gamma(\alpha_i)\Gamma(\beta_i)}
x^{\frac{\alpha_i+\beta_i}{2}-1}
K_{\alpha_i-\beta_i}\!\left(2\sqrt{\alpha_i\beta_i x}\right),\\
\qquad x>0,
\label{eq:gamma_gamma_pdf_final}
\end{IEEEeqnarray}
where $\Gamma(\cdot)$ denotes the Gamma function and $K_\nu(\cdot)$
denotes the modified Bessel function of the second kind. The turbulence strength can be related to the Rytov variance
$\sigma_R^2$, which, \blue{for a plane wave,} is given by \cite{Andrews2005}
\begin{equation}
\sigma_R^2 = 1.23\, C_n^2\, k^{7/6}\, L^{11/6}, 
\end{equation}
where $C_n^2$ denotes the refractive-index structure parameter,
$k=2\pi/\lambda$ is the optical wavenumber, and $L$ is the
propagation distance. The Gamma-Gamma parameters $(\alpha,\beta)$
are functions of $\sigma_R^2$ and characterize the turbulence regime. For a plane wave with negligible inner scale \cite{Andrews2005},
\begin{equation}
\begin{aligned}
\alpha&=\Big[\exp\Big(\tfrac{0.49\,\sigma_R^2}{(1+1.11\,\sigma_R^{12/5})^{7/6}}\Big)-1\Big]^{-1},\\
\beta&=\Big[\exp\Big(\tfrac{0.51\,\sigma_R^2}{(1+0.69\,\sigma_R^{12/5})^{5/6}}\Big)-1\Big]^{-1}.
\end{aligned}
\label{eq:alpha_beta}
\end{equation}

\begin{bl}
Let $F_{I_i}$ denote the distribution function of $I_i$, and let $h_{\ell,i}\in(0,1)$ denote the deterministic link factor of hop $i$.\footnote{\blue{The Gamma-Gamma law describes only the unit-mean fluctuation $I_i$, whereas the deterministic mean path gain (geometric spreading, extinction, receiver efficiency) is a separate factor \cite{Andrews2005,AlHabash2001}. For a receiver aperture of area $\mathcal A$ that is small compared with the beam scale, the transmittance is $\eta\simeq\mathcal A\,I(0)$, with $I(0)$ the instantaneous beam intensity at the aperture center \cite[Eq.~(7)]{Vasylyev2018}. Writing $I(0)=\langle I(0)\rangle I$ with $\langle I\rangle=1$ gives $\eta=h_\ell I$, where $h_\ell=\mathcal A\langle I(0)\rangle$ is the mean aperture transmittance before truncation.}}
Following the truncation-and-renormalization approach adopted for bounded atmospheric PDTs in \cite[Eq.~(8)]{Vasylyev2018}, we condition the scaled Gamma-Gamma irradiance on the physically admissible region $0<h_{\ell,i}I_i<1$. Thus, conditioned on $I_i<1/h_{\ell,i}$, the instantaneous transmissivity is $\eta_i=h_{\ell,i}I_i$, with density
\begin{equation}
f_{\eta_i}(\eta)=\frac{f_{I_i}(\eta/h_{\ell,i})}{h_{\ell,i}\,F_i},\qquad 0<\eta<1,
\label{eq:eta_def}
\end{equation}
where $F_i=F_{I_i}(1/h_{\ell,i})$ is the normalization constant. The application of this construction to the Gamma-Gamma irradiance model constitutes the tractable bounded-PDT approximation adopted in this work.

Conditioned on $\eta_i$, hop $i$ is modeled as the pure-loss bosonic channel $\Lambda_{\eta_i}$, i.e., a beamsplitter of transmissivity $\eta_i$ that couples the transmitted mode to an environmental mode $\hat e$ in the vacuum state \cite{Weedbrook2012,Serafini2017},
\begin{equation}
\hat a_{\mathrm{out}}=\sqrt{\eta_i}\,\hat a_{\mathrm{in}}+\sqrt{1-\eta_i}\,\hat e .
\label{eq:BS}
\end{equation}
This fading-loss description is the standard quantum-optical model of atmospheric links, in which the statistics of the loss are described by a PDT \cite{SemenovVogel2009,Vasylyev2016,Vasylyev2018}. In what follows, $\mathbb E[\cdot]$ denotes expectation with respect to \eqref{eq:eta_def}. We use the quadratures $\hat x=\hat a+\hat a^\dagger$ and $\hat p=-j(\hat a-\hat a^\dagger)$, so that $[\hat x,\hat p]=2j$ and the vacuum state has covariance matrix $\mathbf I_2$, and describe a Gaussian state by its mean vector $\bm m$ and covariance matrix $\Sigma$. Then \eqref{eq:BS} acts as \cite{Weedbrook2012,Serafini2017}
\begin{equation}
\bm m\mapsto\sqrt{\eta}\,\bm m,\qquad \Sigma\mapsto\eta\,\Sigma+(1-\eta)\,\mathbf I_2 .
\label{eq:Gauss_map}
\end{equation}
In particular, a coherent state $|\mu\rangle$ is mapped to the pure coherent state $|\sqrt{\eta}\,\mu\rangle$.

\begin{remark}[Physical interpretation of the channel model]
\label{rem:channel}
Each hop is a pure-loss channel: the field is attenuated by the atmosphere, and the lost fraction $1-\eta_i$ is replaced by the vacuum field of the environment, as in the beamsplitter relation \eqref{eq:BS}. Turbulence enters only through the random transmissivity $\eta_i$, which is known at the receiver; however, knowing $\eta_i$ does not remove the vacuum noise $(1-\eta_i)\mathbf I_2$ added by the loss. We make two further assumptions. First, BPSK signaling and the phase-aligned squeezing of Section~II-D are phase sensitive, so we assume that the receiver tracks the carrier phase and frequency, so that this alignment is preserved at the detector and the residual effect of turbulence is the fading of $\eta_i$. Second, the channel has no excess noise: there is no additive Gaussian noise or background radiation, which is a standard assumption in quantum-limited photodetection \cite{Shapiro2009}. Additive noise would make the received states mixed even for coherent inputs, and its inclusion is left for future work.
\end{remark}
\end{bl}

\subsection{Time-Switching Energy Harvesting at the Relay}
During the EH phase, the relay converts the received optical signal
into electrical energy through a photovoltaic front end. Let
$P_{R,\mathrm{rx}}$ denote the average optical power impinging on the
relay aperture during the harvesting interval. The harvested energy
within one transmission block is then
\begin{equation}
E_{\mathrm{harv}}(t)
=
\eta_{\mathrm{pv}}\, t\, P_{R,\mathrm{rx}},
\label{eq:harvested_energy_final}
\end{equation}
where $\eta_{\mathrm{pv}}\in(0,1]$ denotes the photovoltaic conversion
efficiency.\footnote{A linear EH model is adopted, where the harvested energy is proportional to the incident optical power. This is consistent with quantum photodetection theory, where the measured photocurrent scales with the expectation of the photon number operator under linear attenuation \cite{Glauber1963,Yuen1980}. The model is also widely used in SWIPT systems for analytical tractability and is accurate in the low-to-moderate input power regime \cite{Clerckx2019}.}

\begin{bl}
\textbf{Notation and energy model.} Let $a=|\mu_0|^2$ and $N_2$ denote the mean photon numbers per symbol \emph{transmitted} by the source and by the relay, respectively. The power $P_{R,\mathrm{rx}}$ in \eqref{eq:harvested_energy_final} is the long-term average optical power collected by the relay aperture,
\begin{equation}
P_{R,\mathrm{rx}}=\mathbb E[\eta_1]\,P_S,
\end{equation}
where $P_S$ is the transmitted source power. The power $P_{R,\mathrm{rx}}$ therefore already includes the link factor $h_{\ell,1}$ and the statistics of the first hop. We assume that the EH phase spans many fading coherence intervals, so that the harvested energy, and hence $N_2$, is deterministic, whereas the information symbols of each hop experience block fading that is constant over a symbol and independent across the two hops. Harvesting within a single coherence interval would instead make $N_2$ random and correlated with $\eta_1$, which is left for future work.
\end{bl}

During the forwarding phase, the harvested energy is converted into
optical transmission by an electro-optical modulator with efficiency
$\eta_{\mathrm{el}}\in(0,1]$.
\begin{bl}
Under the half-duplex protocol of Remark~1, the communication phase of duration $1-t$ is shared equally between the S$\to$R and R$\to$D transmissions, so that the relay transmits during a slot of duration $(1-t)/2$. The harvested energy, which is proportional to $t$, is converted into a constant optical transmit power $P_{\mathrm{tx}}=2\eta_{\mathrm{el}}E_{\mathrm{harv}}(t)/(1-t)$ over this slot. With $T_s$ the symbol duration (expressed in the same units as the block duration), the mean photon number per forwarded symbol is $N_2=P_{\mathrm{tx}}T_s/(h\nu)$, that is,
\begin{equation}
\begin{split}
N_2(t)&=\Ct\,\frac{t}{1-t},\\
\Ct&\triangleq\frac{2\eta_{\mathrm{pv}}\eta_{\mathrm{el}}P_{R,\mathrm{rx}}T_s}{h\nu}
=\frac{2\eta_{\mathrm{pv}}\eta_{\mathrm{el}}P_{R,\mathrm{rx}}}{h\nu R_s},
\end{split}
\label{eq:N2_final}
\end{equation}
where $R_s=1/T_s$ is the symbol rate and $h\nu$ denotes the photon energy at optical carrier frequency $\nu$. The parameter $\Ct$ is dimensionless, and $N_2(t)$ is a photon number per forwarded symbol and not the total harvested energy. Increasing $t$ enlarges the relay photon budget, but simultaneously shortens the time available for information forwarding.
\end{bl}

\begin{remark}
The adopted time-switching architecture relies on an energy causality
constraint, whereby the relay harvests optical energy during a
dedicated phase and utilizes this energy for transmission in the
subsequent phase. This model does not require long-term energy storage
or battery operation, but rather assumes short-term energy buffering
over a single transmission block, which is standard in wireless-powered
communication systems~\cite{Ulukus2019, Nasir2013}. Such an assumption
is consistent with practical implementations of EH
receivers employing capacitor or photonic storage elements that
enable immediate reuse of the harvested energy without necessitating
conventional battery storage \cite{Clerckx2019}.
\end{remark}
\subsection{Relay Signaling via DSS}

After decoding the source symbol, the relay re-encodes its decision
into a displaced squeezed state (DSS). DSS signaling combines coherent
displacement and quadrature squeezing, and therefore provides an
additional design degree of freedom beyond conventional coherent-state
relaying.

Let $|0\rangle$ denote the vacuum state of the optical mode. The relay
transmits
\begin{equation}
|\psi_b\rangle
=
D\!\left((-1)^b\gamma_2\right)S(r,\theta_s)|0\rangle,
\qquad b\in\{0,1\},
\label{eq:DSS_general_final}
\end{equation}
where $\gamma_2\ge 0$ denotes the displacement amplitude, $r\ge 0$ is
the squeezing parameter, and $\theta_s\in[0,\pi)$ is the squeezing
phase. The displacement operator is defined as \cite{Chesi2018}
\begin{equation}
D(\alpha)
=
\exp\!\left(\alpha\hat a^\dagger-\alpha^*\hat a\right),
\qquad \alpha\in\mathbb{C},
\label{eq:displacement_operator_final}
\end{equation}
whereas the squeezing operator is given by
\begin{equation}
S(r,\theta_s)
=
\exp\!\left[
\frac{r}{2}
\left(
e^{-j2\theta_s}\hat a^2
-
e^{j2\theta_s}\hat a^{\dagger 2}
\right)
\right].
\label{eq:squeezing_operator_general_final}
\end{equation}

We consider phase-aligned signaling, namely $\theta_s=0$, so that the
squeezing axis is aligned with the displacement quadrature. Under this
choice, \eqref{eq:squeezing_operator_general_final} reduces to
\begin{equation}
S(r)
=
\exp\!\left[
\frac{r}{2}
\left(
\hat a^2-\hat a^{\dagger 2}
\right)
\right],
\label{eq:squeezing_operator_aligned_final}
\end{equation}
and the relay alphabet becomes
\begin{equation}
\mathcal{E}_R
=
\left\{
|\psi_0\rangle,|\psi_1\rangle
\right\}
=
\left\{
D(\gamma_2)S(r)|0\rangle,\,
D(-\gamma_2)S(r)|0\rangle
\right\}.
\label{eq:relay_alphabet_final}
\end{equation}

Let $\hat n=\hat a^\dagger\hat a$ denote the photon-number operator.
The mean photon number of the DSS state in
\eqref{eq:DSS_general_final} is
\begin{equation}
\langle\psi_b|\hat n|\psi_b\rangle
=
\gamma_2^2+\sinh^2(r).
\label{eq:DSS_mean_photons_final}
\end{equation}
Since the relay is energy constrained, the DSS parameters must satisfy
\begin{equation}
\gamma_2^2+\sinh^2(r)=N_2(t).
\label{eq:relay_constraint_final}
\end{equation}
The expression in \eqref{eq:DSS_mean_photons_final} admits a
natural physical interpretation: the total photon budget is
additively decomposed into two components, namely
\begin{itemize}
\item $\gamma_2^2$: the average number of photons associated
with coherent displacement,
\item $\sinh^2(r)$: the average number of photons generated
by the squeezing operation.
\end{itemize}
Hence, DSS signaling enables a controllable redistribution of
the available energy between classical amplitude modulation
(displacement) and quantum noise shaping (squeezing).

To explicitly parameterize the allocation of the relay photon budget
between displacement and squeezing, we introduce the squeezing
fraction \cite{Chesi2018}
\begin{equation}
\zeta
=
\frac{\sinh^2(r)}{N_2(t)},
\qquad 0\le \zeta\blue{{\le}}1.
\label{eq:zeta_final}
\end{equation}
From \eqref{eq:relay_constraint_final} and \eqref{eq:zeta_final}, the
displacement amplitude and squeezing parameter can be written as
\begin{equation}
\gamma_2^2(t,\zeta)
=
(1-\zeta)N_2(t),
\label{eq:gamma2_final}
\end{equation}
and
\begin{equation}
r(t,\zeta)
=
\operatorname{arsinh}\!\big(\sqrt{\zeta N_2(t)}\big).
\label{eq:r_final}
\end{equation}
Hence, for a fixed harvesting fraction $t$, the relay photon budget
$N_2(t)$ is uniquely partitioned between coherent displacement and
quadrature squeezing through the scalar parameter $\zeta$.
\begin{remark}
The relevance of DSS signaling arises from its ability to reshape the quadrature statistics of the transmitted optical field. Specifically, squeezing redistributes the noise across
quadratures by reducing the variance in one quadrature while increasing
it in the conjugate one, whereas displacement preserves the signal-space
separation required for BPSK modulation \cite{Chesi2018}. This additional design degree
of freedom enables improved state discrimination under Helstrom-optimal
detection, particularly in energy-constrained regimes where coherent
signaling becomes suboptimal.
\end{remark}
\subsection{Receiver Processing and DF Operation}

The relay operates according to a TS-DF protocol, and
hence communication involves two consecutive quantum detection
problems.

\textbf{First hop ($S\!\to\!R$):} The relay receives the coherent-state
ensemble
$\{p_b,|\mu_b\rangle\}_{b\in\{0,1\}}$ transmitted by the source and
performs the Helstrom-optimal measurement associated with this binary
ensemble. This measurement minimizes the average symbol error
probability over all physically admissible quantum measurements. Let
$\hat b\in\{0,1\}$ denote the relay decision.

\textbf{Second hop ($R\!\to\!D$):} During the forwarding phase, the
relay re-encodes $\hat b$ into the DSS signaling 
\begin{equation}
|\psi_{\hat b}\rangle
=
D\!\left((-1)^{\hat b}\gamma_2\right)S(r)|0\rangle,
\label{eq:DSS_reencoded_final}
\end{equation}
subject to the energy constraint \eqref{eq:relay_constraint_final}.
The destination then performs the Helstrom-optimal measurement
associated with the received DSS ensemble $\mathcal{E}_R$.

Let $P_{e,1}$ and $P_{e,2}$ denote the symbol error probabilities of
the first and second hops, respectively. Since the relay performs hard
decoding, an error event on the first hop propagates through the
forwarding stage and directly affects the end-to-end reliability.
Consequently, the overall system performance depends jointly on the two
per-hop detection problems.

This completes the physical-layer description of the proposed TS-DF
quantum relay architecture. In the next section, we formalize the
corresponding detection problem and derive the end-to-end performance
metrics under Gamma-Gamma turbulence and
relay energy constraints.
\section{Problem Formulation}
\label{sec:problem}

In this section, we formulate the end-to-end performance criterion of
the proposed TS-DF quantum relay system and state the associated
resource-allocation problem. Since BPSK signaling is employed at both
hops, each receiver faces a binary quantum hypothesis testing problem.
The corresponding optimal performance is characterized by Helstrom's
minimum-error detection principle.

\subsection{Binary Quantum Detection and Helstrom Limit}
\label{subsec:helstrom}

Consider a binary quantum ensemble
\begin{equation}
\mathcal{E}
=
\left\{
p_0,\rho_0;\,
p_1,\rho_1
\right\},
\label{eq:binary_ensemble_final}
\end{equation}
where $\rho_0$ and $\rho_1$ denote the density operators associated
with the two hypotheses and $p_0+p_1=1$ are the prior probabilities.
Any physically admissible receiver is described by a  POVM $\{\Pi_0,\Pi_1\}$ satisfying
\begin{equation}
\Pi_0+\Pi_1=\mathbb{I},
\qquad
\Pi_0\succeq0,\ \Pi_1\succeq0,
\label{eq:binary_POVM_final}
\end{equation}
where $\mathbb{I}$ denotes the identity operator on the corresponding
Hilbert space.

For a given POVM, the average probability of correct detection is
\begin{equation}
P_{\mathrm{c}}
=
\blue{p_0{\Tr[\Pi_0\rho_0]}+p_1{\Tr[\Pi_1\rho_1]}},
\label{eq:Pc_binary_final}
\end{equation}
and hence the associated average probability of error is
\begin{equation}
P_e
=
1-P_{\mathrm{c}}.
\label{eq:Pe_binary_final}
\end{equation}

The minimum achievable error probability over all POVMs satisfying
\eqref{eq:binary_POVM_final} is the Helstrom bound \footnote{
We assume Helstrom-optimal quantum detection, which achieves the
minimum possible discrimination error for binary quantum states.
Although the Helstrom measurement is generally difficult to implement
exactly, several near-optimal receiver architectures based on adaptive
displacement operations and photon counting have been proposed in the
literature~\cite{Takeoka2008, Kennedy1973, Dolinar1973}.
}.
\begin{equation}
P_e^{\Hel}(\mathcal{E})
=
\min_{\{\Pi_0,\Pi_1\}}
\left\{
\blue{1-p_0{\Tr[\Pi_0\rho_0]}-p_1{\Tr[\Pi_1\rho_1]}}
\right\}.
\label{eq:Helstrom_general_final}
\end{equation}

For equiprobable binary signaling, that is,
$p_0=p_1=\tfrac{1}{2}$, \eqref{eq:Helstrom_general_final} admits the
trace-norm representation
\begin{equation}
P_e^{\Hel}(\mathcal{E})
=
\frac{1}{2}
\left(
1-\frac{1}{2}\|\rho_0-\rho_1\|_1
\right),
\label{eq:Helstrom_trace_final}
\end{equation}
where $\|\cdot\|_1$ denotes the trace norm. If the two hypotheses are
pure states, namely
$\rho_0=|\psi_0\rangle\langle\psi_0|$ and
$\rho_1=|\psi_1\rangle\langle\psi_1|$, then
\eqref{eq:Helstrom_trace_final} simplifies to
\begin{equation}
P_e^{\Hel}(\mathcal{E})
=
\frac{1}{2}
\left(
1-\sqrt{1-|\langle\psi_0|\psi_1\rangle|^2}
\right).
\label{eq:Helstrom_pure_final}
\end{equation}
Therefore, in the pure-state case, optimal detection performance is
entirely determined by the overlap between the two candidate states.

\subsection{Per-Hop Detection Errors}
\label{subsec:perhop}

The considered relay architecture induces two binary quantum detection
problems corresponding to the two hops of the link.

\textbf{First hop ($S\!\to\!R$):} The relay receives the coherent-state
BPSK ensemble
\begin{equation}
\mathcal{E}_1
=
\left\{
\frac12,|\mu_0\rangle;\,
\frac12,|\mu_1\rangle
\right\},
\label{eq:E1_final}
\end{equation}
where $|\mu_0\rangle$ and $|\mu_1\rangle$ are given by
\eqref{eq:source_alphabet_final}.
\begin{bl}
Conditioned on the realization $I_1=x<1/h_{\ell,1}$, the transmissivity is $\eta_1(x)=h_{\ell,1}x$ and the received states are the pure coherent states $|\pm\sqrt{\eta_1(x)\,a}\rangle$.
\end{bl}
The corresponding conditional Helstrom error probability is denoted by
\begin{equation}
P_{e,1}^{\Hel}(x)
\triangleq
P_e^{\Hel}\!\left(\mathcal{E}_1\,\middle|\,I_1=x\right).
\label{eq:Pe1_cond_final}
\end{equation}
Averaging over the \blue{truncated} irradiance law \blue{\eqref{eq:eta_def}} yields the first-hop average
Helstrom error probability
\begin{equation}
\bar P_{e,1}^{\Hel}
=
\mathbb{E}
\left[
P_{e,1}^{\Hel}(I_1)
\right].
\label{eq:Pe1_avg_final}
\end{equation}
Since the source photon number $a$ is fixed,
$\bar P_{e,1}^{\Hel}$ is independent of the harvesting fraction $t$.

\medskip

\textbf{Second hop ($R\!\to\!D$):}
\begin{bl}
Conditioned on $I_2=x < 1/h_{\ell,2}$, let $\eta=\eta_2(x)=h_{\ell,2}x$. The destination receives the DSS ensemble
\begin{equation}
\mathcal{E}_2
=
\left\{
\frac12,\rho_0;\,
\frac12,\rho_1
\right\},
\qquad
\rho_b=\Lambda_{\eta}\big(|\psi_b\rangle\langle\psi_b|\big),
\label{eq:E2_final}
\end{equation}
where $|\psi_b\rangle$ are given by \eqref{eq:relay_alphabet_final} and $\Lambda_\eta$ is the pure-loss channel of Remark~\ref{rem:channel}. The relay photon budget depends on the harvesting fraction through \eqref{eq:N2_final}, and the DSS parameters are determined by $(t,\zeta)$ through \eqref{eq:gamma2_final} and \eqref{eq:r_final}. The conditional Helstrom error probability is
\begin{equation}
P_{e,2}^{\Hel}(\eta;t,\zeta)
\triangleq
\frac12\Big(1-\frac12\|\rho_1-\rho_0\|_1\Big),
\label{eq:Pe2_cond_final}
\end{equation}
and the second-hop average Helstrom error probability is
\begin{equation}
\bar P_{e,2}^{\Hel}(t,\zeta)
=
\mathbb{E}
\left[
P_{e,2}^{\Hel}\big(\eta_2;t,\zeta\big)
\right].
\label{eq:Pe2_avg_final}
\end{equation}
\end{bl}
The Helstrom measurement is matched to the instantaneous transmissivity $\eta$, which is assumed known at the receiver (Remark~\ref{rem:channel}); $\bar P_{e,2}^{\Hel}$ below is therefore the average of the conditional optimal errors.
\subsection{End-to-End Reliability}
\label{subsec:e2e_reliability}

Since the relay performs hard decoding and forwards a re-encoded
version of its own decision, the end-to-end symbol error probability
follows the classical DF composition rule. Under the
standard assumption that the two hops are statistically independent,
the average end-to-end error probability is
\begin{equation}
\bar P_e^{\DF}(t,\zeta)
=
\bar P_{e,1}^{\Hel}
\left(
1-\bar P_{e,2}^{\Hel}(t,\zeta)
\right)
+
\left(
1-\bar P_{e,1}^{\Hel}
\right)
\bar P_{e,2}^{\Hel}(t,\zeta).
\label{eq:Pe_DF_final}
\end{equation}
Equivalently,
\begin{equation}
\bar P_e^{\DF}(t,\zeta)
=
\bar P_{e,1}^{\Hel}
+
\bar P_{e,2}^{\Hel}(t,\zeta)
-
2\bar P_{e,1}^{\Hel}\bar P_{e,2}^{\Hel}(t,\zeta).
\label{eq:Pe_DF_affine_final}
\end{equation}

\subsection{Effective Throughput}
\label{subsec:rate}

To jointly capture reliability and the time-switching penalty, we
define the effective throughput, in bits per channel use, as
\begin{bl}
\begin{equation}
R_2(t,\zeta)
=
\frac{1-t}{2}\left(1-\bar P_e^{\DF}(t,\zeta)\right).
\label{eq:rate_final_full}
\end{equation}
Each correctly detected BPSK symbol carries one bit. The factor $(1-t)/2$ is the fraction of each block allocated to the transmissions of each hop under the half-duplex protocol of Remark~1, whereas $1-\bar P_e^{\DF}(t,\zeta)$ represents the probability of correct end-to-end detection.
\end{bl}

The metric in \eqref{eq:rate_final_full} captures the fundamental
trade-off induced by the TS-DF architecture. Increasing the harvesting
fraction $t$ enlarges the relay photon budget $N_2(t)$ and can improve
second-hop discrimination, but it also reduces the available
transmission time. Hence, the optimal design must balance energy
accumulation against forwarding duration.

\subsection{Optimal Resource Allocation}
\label{subsec:opt_problem}

The objective of the proposed TS-DF relay is to determine the
harvesting fraction $t$ and the squeezing allocation parameter $\zeta$
that maximize the effective throughput in
\eqref{eq:rate_final_full}. The resulting optimization problem is
\begin{equation}
\begin{aligned}
[\mathrm{P1}]\qquad
\max_{t,\zeta}\quad
& R_2(t,\zeta)\\
\text{s.t.}\quad
& 0<t<1,\\
& 0\le\zeta\blue{{\le}}1.
\end{aligned}
\label{eq:P1_final_full}
\end{equation}

Problem \eqref{eq:P1_final_full} is a low-dimensional
resource-allocation problem with a clear physical meaning. The variable
$t$ controls the energy-time trade-off through the mapping
$N_2(t)=\blue{\Ct}\frac{t}{1-t}$, whereas $\zeta$ determines how the relay
photon budget is partitioned between coherent displacement and
quadrature squeezing. The resulting optimization therefore couples the
time-switching EH mechanism with the quantum-state design at the relay.

In the following section, we derive the per-hop Helstrom error
probabilities under Gamma-Gamma turbulence and
establish structural \blue{and asymptotic} properties of the throughput function
\blue{$R_2(t,\zeta)$} that enable efficient solution of
\eqref{eq:P1_final_full}.
\vspace{-2ex}

\section{\blue{Helstrom Error Analysis and Rate Optimization}}
\label{sec:analysis}

\begin{bl}
This section characterizes the reliability and throughput of the proposed TS-DF quantum relay. We first derive the first-hop error and the exact mixed-state description of the second hop. We then give a closed-form fidelity with two-sided bounds on the second-hop Helstrom error, establish the existence of the loss-aware photon partition, and derive an exact series representation for coherent relaying and for a CP benchmark. Finally, we characterize the asymptotic reliability and the structure of the effective throughput. Each lemma, proposition and theorem is proved in a dedicated appendix. Throughout, we define
\begin{equation}
\vphi(u)\triangleq\tfrac12\big(1-\sqrt{1-e^{-u}}\big),\qquad u\ge 0,
\label{eq:phi_def}
\end{equation}
which is continuous and strictly decreasing, with $\vphi(0)=\tfrac12$, $\vphi(\infty)=0$ and $\vphi(u)\le\tfrac12e^{-u}$.
\end{bl}

\subsection{\blue{First-Hop Helstrom Error}}
\label{subsec:firsthop_analysis}

\begin{bl}
By \eqref{eq:source_alphabet_final}, the source transmits the BPSK ensemble $|\mu_0\rangle=|+\sqrt a\rangle$, $|\mu_1\rangle=|-\sqrt a\rangle$, where $a=|\mu_0|^2$ is the \emph{transmitted} mean photon number, which is distinct from the irradiance realization $x=I_1$. Conditioned on $I_1=x < 1/h_{\ell,1}$, the received states are the pure coherent states
\begin{equation}
|\psi_0^{(1)}(x)\rangle=|\sqrt{\eta_1(x)}\,\mu_0\rangle,\qquad
|\psi_1^{(1)}(x)\rangle=|\sqrt{\eta_1(x)}\,\mu_1\rangle,
\label{eq:firsthop_states_final}
\end{equation}
so that the received amplitude never exceeds the transmitted one. Using $\langle\alpha|\beta\rangle=\exp(-|\alpha|^2/2-|\beta|^2/2+\alpha^*\beta)$ and $|\mu_0-\mu_1|^2=4a$, the squared overlap is $\exp(-4a\eta_1(x))$, and by \eqref{eq:Helstrom_pure_final} the conditional and average first-hop Helstrom errors are
\begin{equation}
P_{e,1}^{\Hel}(x)=\vphi\big(4a\,\eta_1(x)\big),
\qquad
\bar P_{e,1}^{\Hel}=\mathbb E\big[\vphi(4a\,\eta_1)\big].
\label{eq:firsthop_average_error_final}
\end{equation}
As expected, $\bar P_{e,1}^{\Hel}$ does not depend on the harvesting fraction $t$. A series representation of \eqref{eq:firsthop_average_error_final} is given in Section~\ref{subsec:series_representation}.
\end{bl}

\subsection{\blue{Second Hop: Exact Mixed-State Formulation}}
\label{subsec:secondhop_analysis}

\begin{bl}
The following lemma gives the exact states received on the second hop.

\begin{lemma}[Attenuator output for DSS inputs]
\label{lem:output}
Let $\gamma_2\ge0$, $r\ge0$, and $\eta\in[0,1]$. The states $\rho_b=\Lambda_\eta(|\psi_b\rangle\langle\psi_b|)$ are Gaussian with means $\bm m_b=\big((-1)^b\,2\sqrt{\eta}\,\gamma_2,\,0\big)^{\!\top}$ and common covariance matrix
\begin{equation}
\begin{aligned}
\Sigma&=\mathrm{diag}(v_x,v_p),\\
v_x&=1-\eta+\eta e^{-2r},\quad
v_p=1-\eta+\eta e^{2r}.
\end{aligned}
\label{eq:Sigma_out}
\end{equation}
Moreover, $\det\Sigma=1+4\eta(1-\eta)\sinh^2 r$ and $\Tr\rho_b^2=(\det\Sigma)^{-1/2}$. Hence $\rho_b$ is pure if and only if $r=0$ or $\eta\in\{0,1\}$, and it is strictly mixed for $r>0$ and $0<\eta<1$. For $r=0$ (coherent relaying), $\rho_b=|(-1)^b\sqrt{\eta N_2}\rangle\langle(-1)^b\sqrt{\eta N_2}|$.
\end{lemma}

\textit{Proof:} See Appendix~\ref{app:proof_output}.

\begin{remark}[Why DSS states are mixed]
\label{rem:why_mixed}
The loss couples the signal mode to an environment mode in the vacuum state. After the environment is discarded, the output covariance is $\eta\Sigma_{\mathrm{in}}+(1-\eta)\mathbf I_2$, as in \eqref{eq:Gauss_map}, i.e., the vacuum covariance $\mathbf I_2$ is mixed in with weight $1-\eta$. For a coherent input, $\Sigma_{\mathrm{in}}=\mathbf I_2$, so the added vacuum noise equals the noise already present and the output is the pure coherent state $|\sqrt\eta\,\mu\rangle$. For a squeezed input, $\Sigma_{\mathrm{in}}=\mathrm{diag}(e^{-2r},e^{2r})\neq\mathbf I_2$, so the vacuum noise does not match the input noise: it raises the squeezed-quadrature variance from $e^{-2r}$ to $v_x=1-\eta(1-e^{-2r})$ and gives $\det\Sigma=1+4\eta(1-\eta)\sinh^2r>1$. Hence $\mathrm{Tr}\rho^2=(\det\Sigma)^{-1/2}<1$ for $r>0$ and $0<\eta<1$, i.e., the state is mixed because the loss correlates it with the environment.  The behavior of the DSS hop in Sections~\ref{subsec:fidelity} to~\ref{subsec:asymptotics} follows from these two facts.
\end{remark}
\end{bl}

\subsection{\blue{Closed-Form Fidelity and Two-Sided Bounds}}
\label{subsec:fidelity}

\begin{bl}
Let $F(\rho_0,\rho_1)=\|\sqrt{\rho_0}\sqrt{\rho_1}\|_1^2$ denote the squared Uhlmann fidelity \cite{Uhlmann1976,Jozsa1994}, which equals $|\langle\psi_0|\psi_1\rangle|^2$ for pure states, and define
\begin{equation}
\Theta_{\mathrm{in}}\triangleq4\gamma_2^2e^{2r}=4(1-\zeta)N_2\,e^{2r}.
\label{eq:Theta_in}
\end{equation}

\begin{proposition}[Gaussian fidelity and two-sided bounds]
\label{prop:fidelity}
For the states of Lemma~\ref{lem:output}, we have 
\begin{equation}
\begin{aligned}
F&=e^{-\Th},\\
\Th(\eta;t,\zeta)&=\frac{4\eta\gamma_2^2}{1-\eta+\eta e^{-2r}}
=\frac{4\eta(1-\zeta)N_2}{1-\eta+\eta e^{-2r}},
\end{aligned}
\label{eq:Fclosed}
\end{equation}
and the Helstrom error satisfies
\begin{equation}
\vphi(\Th)\;\le\;P_{e,2}^{\Hel}(\eta;t,\zeta)\;\le\;\tfrac12e^{-\Th/2}.
\label{eq:FvG}
\end{equation}
The lower bound holds with equality whenever $\rho_0,\rho_1$ are pure, i.e., for $\eta=1$ or $r=0$. In particular, $P_{e,2}^{\Hel}=\vphi(\Theta_{\mathrm{in}})$ for $\eta=1$, and $P_{e,2}^{\Hel}=\vphi(4\eta N_2)$ for $r=0$. Moreover,
\begin{equation}
\begin{aligned}
4\eta(1-\zeta)N_2\;&\le\;\Th\;\le\;\eta\,\Theta_{\mathrm{in}},\\
\Th\;&\le\;\frac{4\eta(1-\zeta)N_2}{1-\eta}\quad(0\le\eta<1),
\end{aligned}
\label{eq:Thbounds}
\end{equation}
and, for every $\eta\in[0,1]$,
\begin{equation}
P_{e,2}^{\Hel}(\eta;t,\zeta)\;\ge\;\vphi(\Theta_{\mathrm{in}}).
\label{eq:dataproc}
\end{equation}
\end{proposition}

\textit{Proof:} See Appendix~\ref{app:proof_fidelity}.

\begin{remark}
The fidelity depends on $\Sigma$ only through $v_x$, and is therefore insensitive to the noise added to the anti-squeezed quadrature. The bounds \eqref{eq:FvG} are thus rigorous but not tight for mixed states. They are used for the asymptotic analysis of Section~\ref{subsec:asymptotics}, whereas the design of the squeezing allocation relies on the exact trace-norm error.
\end{remark}
\end{bl}

\subsection{\blue{Loss-Aware Photon Allocation}}
\label{subsec:allocation}

\begin{bl}
For the lossless case $\eta=1$, the partition that maximizes $\Theta_{\mathrm{in}}$ at a fixed photon budget, and hence minimizes the Helstrom error, is the known result
\begin{equation}
\zeta_{\mathrm L}(t)=\frac{N_2(t)}{2N_2(t)+1},
\label{eq:zeta_lossless}
\end{equation}
established in \cite{Chesi2018, Paris2001, Walsh2025} (a short derivation is included in Appendix~\ref{app:proof_bound}). At $\zeta_{\mathrm L}$, one has $e^{2r}=2N_2+1$, $\gamma_2^2=N_2(N_2+1)/(2N_2+1)$, and $\Theta_{\mathrm{in}}=4N_2(N_2+1)$. Under the fading-loss channel of Section~II-B, the optimal partition no longer takes this form, and it is defined as follows.

\begin{proposition}[Existence of a loss-aware allocation]
\label{prop:allocation}
For every $t\in(0,1)$, the map $\zeta\mapsto\bar P_{e,2}^{\Hel}(t,\zeta)$ is continuous on $[0,1]$. Hence the set of minimizers is non-empty and closed, and
\begin{equation}
\zeta^\star(t)\triangleq\min\arg\min_{\zeta\in[0,1]}\bar P_{e,2}^{\Hel}(t,\zeta)
\label{eq:zeta_star_def}
\end{equation}
exists. We write $\bar P_{e,2}^{\star}(t)=\bar P_{e,2}^{\Hel}(t,\zeta^\star(t))$.
\end{proposition}

\textit{Proof:} See Appendix~\ref{app:proof_allocation}.

We do not claim uniqueness, and $\zeta^\star(t)$ is not available in closed form: the error \eqref{eq:Pe2_cond_final} is a trace-norm functional of mixed Gaussian states, which Proposition~\ref{prop:fidelity} brackets without determining it. The allocation $\zeta^\star(t)$ is computed by a one-dimensional numerical minimization of the exact averaged error (Section~\ref{sec:numerical}). 
\end{bl}

\subsection{\blue{Exact Series Representation}}
\label{subsec:series_representation}

\begin{bl}
Let $\mathcal L_{I}(s)\triangleq\int_0^\infty e^{-sx}f_I(x)\,dx$ denote the Laplace transform of the (untruncated) Gamma-Gamma irradiance, and let $b_k=(-1)^{k+1}\binom{1/2}{k}>0$ be the coefficients of the binomial series $1-\sqrt{1-z}=\sum_{k\ge1}b_kz^k$ \cite[\S1.2(iv)]{DLMF}. Then $\vphi(u)=\tfrac12\sum_{k\ge1}b_ke^{-ku}$ for $u>0$ and $\sum_{k\ge1}b_k=1$. For the truncated density, the mass of the untruncated law above $1/h_{\ell,i}$ enters through the nonnegative terms
\begin{IEEEeqnarray}{cCl}
\nonumber
\tau_i(s)&\triangleq&\frac{1}{F_i}\int_{1/h_{\ell,i}}^{\infty}e^{-sh_{\ell,i}x}f_{I_i}(x)\,dx,\qquad \\
\tau_i^{\vphi}(\Theta) &\triangleq& \frac{1}{F_i}\int_{1/h_{\ell,i}}^{\infty}\vphi(\Theta h_{\ell,i}x)\,f_{I_i}(x)\,dx .
\label{eq:tau_def}
\end{IEEEeqnarray}

\begin{proposition}[Exact Laplace-series representation]
\label{prop:series_laplace}
For every $s\ge0$, we have
\begin{equation}
\mathbb E[e^{-s\eta_i}]=\frac{1}{F_i}\mathcal L_{I_i}(h_{\ell,i}s)-\tau_i(s),\qquad 0\le\tau_i(s)\le\frac{1-F_i}{F_i}\,e^{-s}.
\label{eq:Laplace_trunc}
\end{equation}
Consequently, for every $\Theta\ge0$,
\begin{equation}
\mathbb E\big[\vphi(\Theta\,\eta_i)\big]
=\frac{1}{2F_i}\sum_{k=1}^{\infty}b_k\,\mathcal L_{I_i}\big(k\,h_{\ell,i}\Theta\big)-\tau_i^{\vphi}(\Theta),
\label{eq:series_general}
\end{equation}
with $0\le\tau_i^{\vphi}(\Theta)\le\frac{1-F_i}{F_i}\vphi(\Theta)$. In particular, the first-hop error and the coherent second-hop error ($\zeta=0$) are
\begin{align}
\bar P_{e,1}^{\Hel}&=\frac{1}{2F_1}\sum_{k\ge1}b_k\,\mathcal L_{I_1}(4ak\,h_{\ell,1})-\tau_1^{\vphi}(4a),\label{eq:Pe1_series_final}\\
\bar P_{e,2}^{\mathrm{coh}}(t)&=\frac{1}{2F_2}\sum_{k\ge1}b_k\,\mathcal L_{I_2}(4kN_2(t)\,h_{\ell,2})-\tau_2^{\vphi}(4N_2(t)).\label{eq:Pe_coh_final}
\end{align}
\end{proposition}

\textit{Proof:} See Appendix~\ref{app:proof_prop_series}.

The series in \eqref{eq:series_general} without the correction $\tau_i^{\vphi}$ is thus an upper bound that lies within $\frac{1-F_i}{F_i}\vphi(\Theta)$ of the exact value. Since $\mathcal L_{I_i}(h_{\ell,i}k\Theta)$ decays only polynomially in $\Theta$ (Lemma~\ref{lem:laplace_asympt}), whereas $\tau_i^{\vphi}(\Theta)=O(e^{-\Theta})$, the relative effect of the correction vanishes as $\Theta\to\infty$. Since the states are pure for $\zeta=0$, $\bar P_{e,2}^{\mathrm{coh}}(t)=\bar P_{e,2}^{\Hel}(t,0)$; Theorem~\ref{thm:asymptotics} writes the latter as a function of $N_2$.
\end{bl}

\subsection{\blue{Covariance-Preserving Benchmark}}
\label{subsec:cp_benchmark}

\begin{bl}
To quantify the effect of the vacuum noise entering the loss port, we compare the physical DSS hop with a  channel that attenuates the displacement but leaves the covariance matrix of the lossless state unchanged. We call it the CP benchmark,
\begin{equation}
\begin{aligned}
P_{e,2}^{\mathrm{cp}}(\eta;t,\zeta)&\triangleq\vphi\big(\eta\,\Theta_{\mathrm{in}}(t,\zeta)\big),\\
\bar P_{e,2}^{\mathrm{cp}}(t,\zeta)&=\mathbb E\big[P_{e,2}^{\mathrm{cp}}(\eta_2;t,\zeta)\big].
\end{aligned}
\label{eq:cp_def}
\end{equation}
It coincides with the exact error for $\eta=1$ or $r=0$, but it is {not} the output of the pure-loss channel \eqref{eq:BS}, since it ignores the vacuum noise $(1-\eta)\mathbf I_2$. It is used only as a benchmark. By \eqref{eq:Thbounds}, $\Th\le\eta\Theta_{\mathrm{in}}$, hence
\begin{equation}
P_{e,2}^{\mathrm{cp}}(\eta;t,\zeta)\le\vphi(\Th)\le P_{e,2}^{\Hel}(\eta;t,\zeta)
\label{eq:cp_lower}
\end{equation}
for all $\eta,\zeta$. By \eqref{eq:series_general}, $\bar P_{e,2}^{\mathrm{cp}}(t,\zeta)=\frac{1}{2F_2}\sum_kb_k\mathcal L_{I_2}(kh_{\ell,2}\Theta_{\mathrm{in}})-\tau_2^{\vphi}(\Theta_{\mathrm{in}})$, which, at $\zeta_{\mathrm L}$, depends on the photon budget only through $\Theta_{\mathrm{in}}=4N_2(N_2+1)$.  For the considered DSS ensemble, covariance preservation could be realized through an attenuator coupled to a squeezed-vacuum environment whose covariance matches that of the input states. Such a configuration, however, requires an additional quantum resource and an environment matched to the transmitted squeezing parameter, which are not available in the proposed EH relay architecture. Hence, the CP model is considered solely as an idealized benchmark.
\end{bl}

\subsection{\blue{Asymptotic Reliability under Loss}}
\label{subsec:asymptotics}

\begin{bl}
Let $g_2\triangleq\min\{\alpha_2,\beta_2\}$ denote the turbulence-induced diversity order of the second hop, and assume $\alpha_2\ne\beta_2$ (the case $\alpha_2=\beta_2$ adds a logarithmic correction and is excluded).
\end{bl}

\begin{bl}\begin{lemma}[Gamma-Gamma Laplace Asymptotics]
\label{lem:laplace_asympt}
Let $I$ be Gamma-Gamma distributed with parameters $(\alpha,\beta)$, $\alpha\ne\beta$,
and let $g=\min\{\alpha,\beta\}$. Then the density satisfies $f_I(x)\sim c_0x^{g-1}$ as $x\to0^+$, with
$c_0=\Gamma(|\alpha-\beta|)(\alpha\beta)^g/[\Gamma(\alpha)\Gamma(\beta)]$, and its Laplace transform satisfies
\begin{equation}
\mathcal{L}_I(s)
=
c_\infty s^{-g}\big(1+o(1)\big),
\qquad s\to\infty,\qquad c_\infty=c_0\Gamma(g).
\label{eq:laplace_asympt_final}
\end{equation}
\end{lemma}\end{bl}
\begin{IEEEproof}
    See Appendix~\ref{app:proof_gamma_gamma_asympt}.
\end{IEEEproof}

\begin{bl}
Let $h=h_{\ell,2}$ and define
\begin{equation}
K\triangleq c_0\int_0^\infty u^{g_2-1}\vphi(u)\,du=\frac{c_\infty}{2}\sum_{k=1}^\infty b_k\,k^{-g_2},\qquad \Kb\triangleq\frac{K}{F_2}.
\label{eq:K_def}
\end{equation}

\begin{lemma}[Asymptotics of the averaged error]
\label{lem:avg_asympt}
For the second-hop transmissivity $\eta_2$ of \eqref{eq:eta_def}, we have
\begin{equation}
\mathbb E\big[\vphi(\Theta\,\eta_2)\big]=\Kb\,(h\Theta)^{-g_2}\big(1+o(1)\big),\qquad\Theta\to\infty.
\label{eq:avg_asympt}
\end{equation}
\end{lemma}

\textit{Proof:} See Appendix~\ref{app:proof_avg_asympt}.

\begin{theorem}[Asymptotic reliability under physical loss]
\label{thm:asymptotics}
As $N_2\to\infty$:
\begin{enumerate}[label=(\alph*),leftmargin=2em]
\item $\bar P_{e,2}^{\star}(N_2)\sim\bar P_{e,2}^{\Hel}(N_2,0)\sim \Kb\,(4hN_2)^{-g_2}$. That is, even optimally chosen squeezing changes neither the turbulence-induced diversity order $g_2$ of coherent relaying nor its leading constant $\Kb$;
\item every minimizer $\zeta^\star(N_2)$ satisfies $\zeta^\star(N_2)\to0$.
\end{enumerate}
\end{theorem}

\textit{Proof:} See Appendix~\ref{app:proof_asymptotics}.

\begin{corollary}[CP benchmark exponent]
\label{cor:benchmark}
For the CP benchmark at the lossless partition, we have
\begin{equation}
\begin{aligned}
\bar P_{e,2}^{\mathrm{cp}}\big(t,\zeta_{\mathrm L}(t)\big)&\sim \Kb\,\big(4hN_2(N_2+1)\big)^{-g_2}\\
&\sim \Kb(4hN_2^2)^{-g_2},\qquad N_2\to\infty.
\end{aligned}
\end{equation}
\end{corollary}
\begin{IEEEproof}
Apply Lemma~\ref{lem:avg_asympt} with $\Theta=\Theta_{\mathrm{in}}(t,\zeta_{\mathrm L})=4N_2(N_2+1)$ in \eqref{eq:cp_def}.
\end{IEEEproof}

\begin{remark}
Theorem~\ref{thm:asymptotics} and Corollary~\ref{cor:benchmark} together locate the effect of the loss. If squeezing were preserved, the diversity order would double ($2g_2$ instead of $g_2$), because $e^{2r}$ grows linearly with $N_2$. The vacuum noise entering the loss port removes this growth: since $v_x\ge1-\eta$, the exponent $\Th$ in \eqref{eq:Thbounds} is capped at a value linear in $N_2$ for all $\eta<1$. The benefit of squeezing is therefore a finite-budget effect, quantified in Section~\ref{sec:numerical}.
\end{remark}
\end{bl}

\subsection{\blue{End-to-End Error Probability and Effective Throughput}}
\label{subsec:e2e_df_rate}
With the per-hop errors \eqref{eq:firsthop_average_error_final} and \eqref{eq:Pe2_avg_final}, the end-to-end error is given by \eqref{eq:Pe_DF_affine_final}. For fixed $\bar P_{e,1}^{\Hel}\in(0,\frac12)$,
\eqref{eq:Pe_DF_affine_final} is affine and strictly increasing in
$\bar P_{e,2}^{\Hel}(t,\zeta)$. Therefore, any reduction in the
second-hop Helstrom error directly improves end-to-end reliability.

\begin{bl}
Using \eqref{eq:rate_final_full}, the effective throughput is
\begin{equation}
R_2(t,\zeta)
=
\frac{1-t}{2}
\Big(
1-\bar P_{e,1}^{\Hel}
-
(1-2\bar P_{e,1}^{\Hel})\bar P_{e,2}^{\Hel}(t,\zeta)
\Big).
\label{eq:R2_affine_final}
\end{equation}
Since $1-2\bar P_{e,1}^{\Hel}>0$ (because $\vphi(u)<\frac12$ for $u>0$ and $\eta_1>0$ almost surely), for every $t$ the maximization over $\zeta$ is attained at $\zeta^\star(t)$ of Proposition~\ref{prop:allocation}, and [P1] reduces to the scalar problem
\begin{equation}
t^\star\in\arg\max_{0<t<1}R_2\big(t,\zeta^\star(t)\big).
\label{eq:tstar_scalar_final}
\end{equation}

\begin{proposition}[Interior optimal harvesting fraction]
\label{prop:interior}
The function $R_2^\star(t)\triangleq R_2(t,\zeta^\star(t))$ is continuous on $(0,1)$, $R_2^\star(t)\to\frac14$ as $t\to0^+$, $R_2^\star(t)\to0$ as $t\to1^-$, and $R_2^\star(t)>\frac14$ for all sufficiently small $t>0$. Consequently, the maximum in \eqref{eq:tstar_scalar_final} exists and is attained at some $t^\star\in(0,1)$.
\end{proposition}

\textit{Proof:} See Appendix~\ref{app:proof_interior}.

Proposition~\ref{prop:interior} guarantees an interior operating point. We do not claim uniqueness, since the concavity argument available for a smooth closed-form allocation does not carry over to a numerical minimizer of a trace-norm functional; the throughput is unimodal in all our numerical experiments (Section~\ref{sec:numerical}).

\begin{proposition}[CP throughput bound]
\label{prop:bound}
Let $R_2^{\mathrm{cp}}(t)$ denote the effective throughput \eqref{eq:R2_affine_final} with $\bar P_{e,2}^{\Hel}$ replaced by $\min_\zeta\bar P_{e,2}^{\mathrm{cp}}(t,\zeta)$, and let $R_2^\star(t)$ be as in Proposition~\ref{prop:interior}. Then:
\begin{enumerate}[label=(\roman*),leftmargin=2em]
\item for every $t\in(0,1)$, $R_2^{\mathrm{cp}}(t)\ge R_2^{\star}(t)$, hence $\max_tR_2^{\mathrm{cp}}(t)\ge\max_tR_2^{\star}(t)$;
\item as $N_2\to\infty$ with the other parameters fixed (e.g., through $\Ct\to\infty$),
\begin{equation}
R_2^{\mathrm{cp}}(t)-R_2^{\star}(t)\sim\frac{1-t}{2}\big(1-2\bar P_{e,1}^{\Hel}\big)\Kb\,(4hN_2)^{-g_2}\longrightarrow0 .
\end{equation}
\end{enumerate}
\end{proposition}

\textit{Proof:} See Appendix~\ref{app:proof_bound}.

Hence the CP benchmark, which admits the exact series representation \eqref{eq:series_general} and the closed-form exponent $4N_2(N_2+1)$ at $\zeta_{\mathrm L}$, is a rigorous analytical upper bound on the achievable throughput, whatever the squeezing fraction. Its absolute gap vanishes at the same polynomial rate as the coherent second-hop error itself, so that $R_2^{\mathrm{cp}}/R_2^{\star}\to1$; in relative terms of the second-hop error, however, the benchmark is not tight, since $\bar P_{e,2}^{\mathrm{cp}}=O(N_2^{-2g_2})$ whereas $\bar P_{e,2}^{\star}=\Theta(N_2^{-g_2})$.
\end{bl}

\section{\blue{Numerical Results and Design Insights}}
\label{sec:numerical}

\begin{bl}
In this section, we illustrate the analytical results of Section~\ref{sec:analysis} and the main design insights. We examine: i) the accuracy of the closed-form fidelity and bounds, ii) the loss-aware squeezing allocation compared with the lossless rule, iii) the regime in which squeezing is beneficial, iv) the asymptotic behavior of Theorem~\ref{thm:asymptotics}, and v) the throughput-optimal harvesting design and the quality of the CP bound. The simulation parameters are given in Table~\ref{tab:params}. Coherent and CP errors follow from the series of Proposition~\ref{prop:series_laplace}. For DSS, the output states of Lemma~\ref{lem:output} are built in a truncated Fock basis of dimension $M=\lfloor\bar n+9\sqrt{\bar n+1}+45\rfloor$, with $\bar n=\gamma_2^2+\sinh^2 r$, and the conditional Helstrom error is $\tfrac12\big(1-\tfrac12\|\rho_0-\rho_1\|_1\big)$, where the trace norm is the sum of the absolute eigenvalues of $\rho_0-\rho_1$. Increasing $M$ by $60$ changes the conditional error by less than $10^{-7}$ in relative terms. The optimizers $\zeta^\star(t)$ and $t^\star$ use bounded scalar minimization with tolerances $10^{-7}$ and $10^{-5}$, respectively. Monte Carlo results use $10^5$ realizations of the truncated law \eqref{eq:eta_def}.

 Fig.~\ref{fig:prop2} plots the exact conditional Helstrom error $P_{e,2}^{\Hel}(\eta)$ for coherent ($\zeta=0$) and DSS ($\zeta=0.3$) signaling, together with the bounds \eqref{eq:FvG} and the CP error. Three observations follow. First, for $\zeta=0$ the states are pure (Lemma~\ref{lem:output}), so the lower bound coincides with the exact error, which is the equality case stated in Proposition~\ref{prop:fidelity}. Second, for $\zeta=0.3$ the states are mixed for every $\eta\in(0,1)$ and the exact error lies between the two bounds, but the lower bound is loose: the fidelity ignores the noise in the anti-squeezed quadrature, which is why the bounds are used for the asymptotic analysis and not as error estimates. Third, the DSS error is smaller than the coherent error only for sufficiently large $\eta$ (here $\eta\gtrsim0.6$); under strong loss, the vacuum noise $1-\eta$ in $v_x$ cancels the squeezing gain (Remark~\ref{rem:why_mixed}) and the DSS error exceeds the coherent one. The CP error lies below the exact error, in agreement with \eqref{eq:cp_lower}: the two coincide as $\eta\to1$ and both tend to $1/2$ as $\eta\to0$, while at intermediate $\eta$ the CP error is orders of magnitude smaller. It is therefore a benchmark and not an achievable performance. We also verified numerically that the closed-form fidelity \eqref{eq:Fclosed} coincides with $\|\sqrt{\rho_0}\sqrt{\rho_1}\|_1^2$ over a wide grid of photon budgets, squeezing fractions and transmissivities.

\begin{figure}[t]
\centering
\includegraphics[width=\columnwidth]{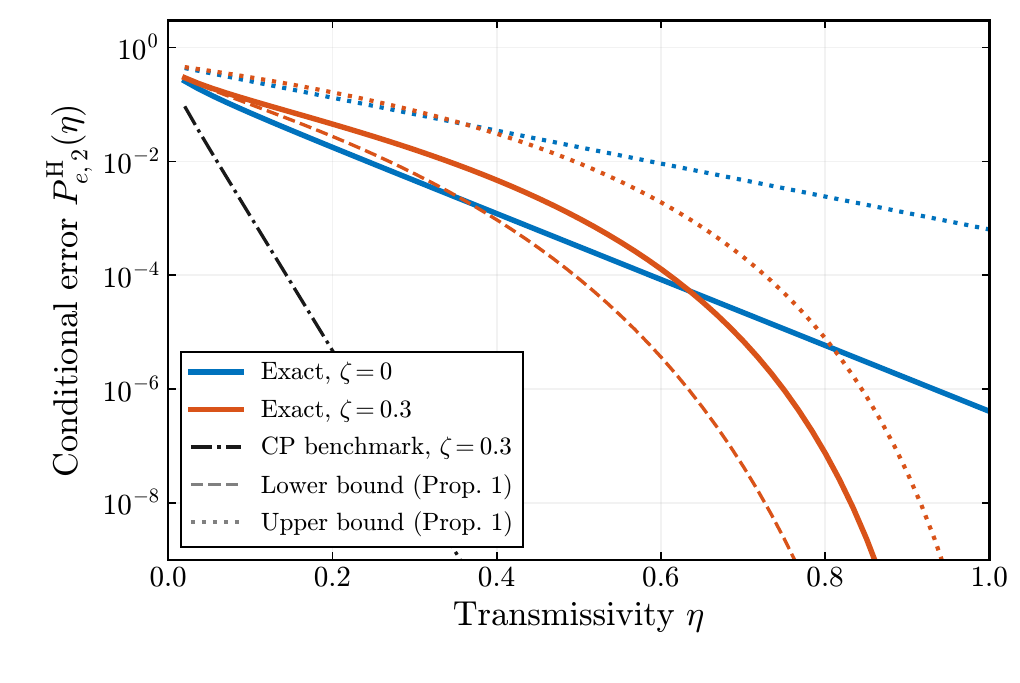}
\caption{Validation of Proposition~\ref{prop:fidelity}: conditional Helstrom error $P_{e,2}^{\Hel}(\eta)$ versus the transmissivity $\eta$ at $N_2=3.33$, with the bounds \eqref{eq:FvG} and the CP benchmark.}
\label{fig:prop2}
\end{figure}

Fig.~\ref{fig:zeta}(a) plots $\bar P_{e,2}^{\Hel}(t,\zeta)$ versus $\zeta$ for three harvesting fractions; the Monte Carlo points confirm the exact evaluation. As guaranteed by Proposition~\ref{prop:allocation}, a minimizer exists, and it is attained at a small squeezing fraction $\zeta^\star(t)$, far below the lossless rule $\zeta_{\mathrm L}(t)$ of \eqref{eq:zeta_lossless}. Beyond $\zeta^\star$ the error increases with $\zeta$ and, at $\zeta_{\mathrm L}$, it exceeds that of coherent relaying: moving photons from displacement to squeezing reduces $\gamma_2$, while the gain from squeezing is capped by $v_x\ge1-\eta$ (Remark~\ref{rem:why_mixed}). The CP benchmark at $\zeta_{\mathrm L}$ lies below the exact error, by a factor that grows with $N_2$, which shows how much of the nominal squeezing gain the vacuum noise removes. Fig.~\ref{fig:zeta}(b) zooms on small $\zeta$: squeezing improves on coherent relaying only in a narrow interval around $\zeta^\star(t)$.

\begin{figure}[t]
\centering
\includegraphics[width=\columnwidth]{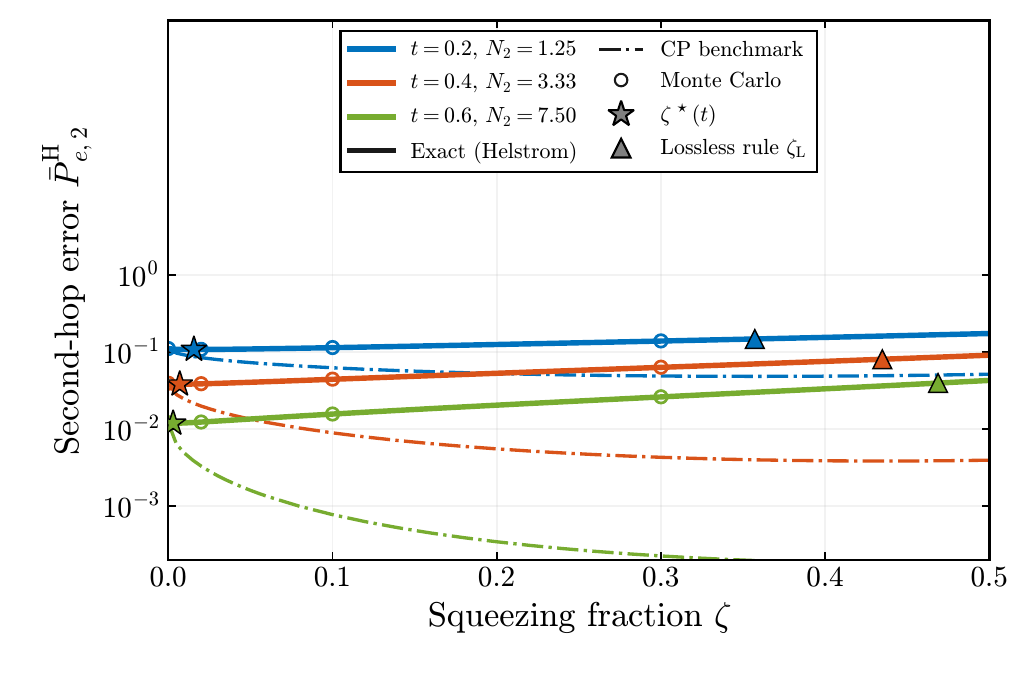}\\[-0.5mm]{\small (a)}\\[1mm]
\includegraphics[width=\columnwidth]{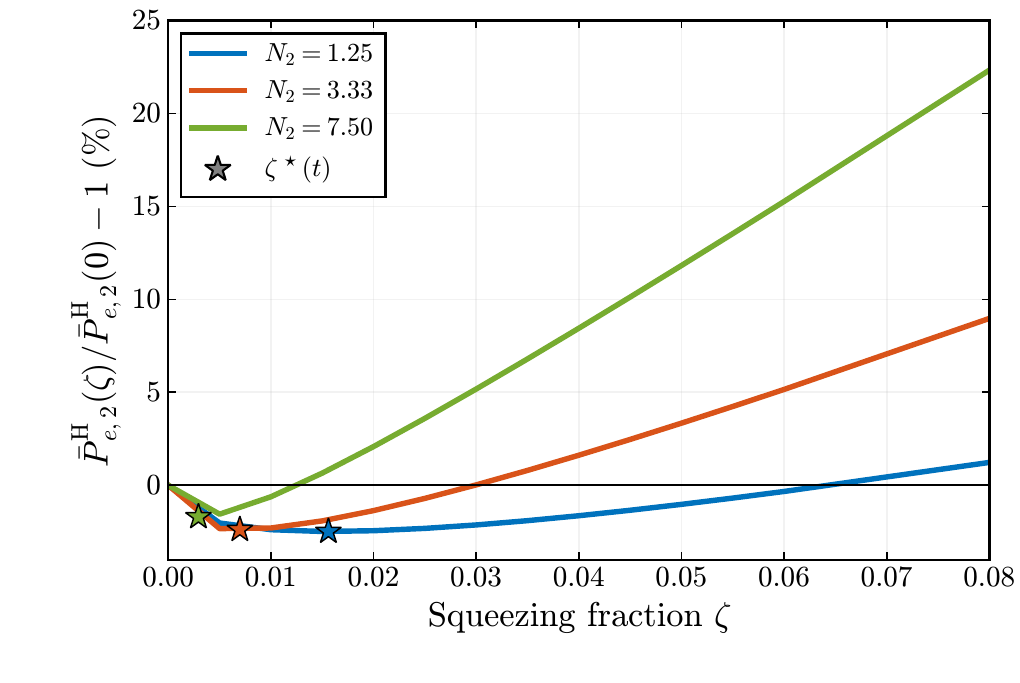}\\[-0.5mm]{\small (b)}
\caption{Loss-aware allocation. (a) Second-hop error $\bar P_{e,2}^{\Hel}(t,\zeta)$ versus $\zeta$: exact Helstrom error (solid), CP benchmark (dash-dotted), Monte Carlo (circles), loss-aware optimum $\zeta^\star(t)$ (stars) and lossless rule $\zeta_{\mathrm L}$ (triangles). (b) Relative change with respect to coherent relaying ($\zeta=0$).}
\label{fig:zeta}
\end{figure}

Fig.~\ref{fig:regime}(a) shows that $\zeta^\star$ stays small and decays toward zero as $N_2$ grows, in agreement with Theorem~\ref{thm:asymptotics}(b), whereas the lossless rule \eqref{eq:zeta_lossless} tends to $1/2$. Squeezing is therefore useful only at moderate budgets and mild loss. Fig.~\ref{fig:regime}(b) plots the relative error reduction $1-\bar P_{e,2}^{\star}/\bar P_{e,2}^{\mathrm{coh}}$ for $h_{\ell,2}=0.6$ and three turbulence strengths. The reduction is modest for the turbulence of Table~\ref{tab:params}, but reaches 10--30\% in weak turbulence. Strong turbulence concentrates the averaged error on deep fades, where $v_x\approx1$ and squeezing is ineffective (Remark~\ref{rem:why_mixed}); weak turbulence concentrates $\eta_2$ around $h_{\ell,2}$, where part of the squeezing gain survives. In all cases, the reduction peaks at moderate $N_2$ and vanishes as $N_2$ grows, in agreement with Theorem~\ref{thm:asymptotics}.

\begin{figure}[t]
\centering
\includegraphics[width=\columnwidth]{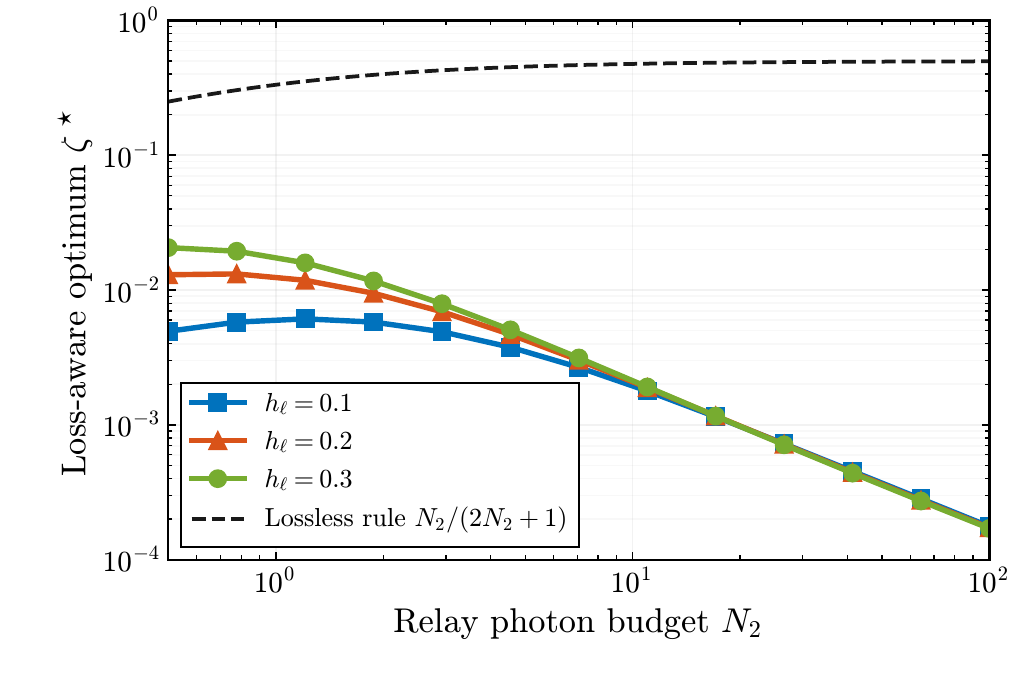}\\[-0.5mm]{\small (a)}\\[1mm]
\includegraphics[width=\columnwidth]{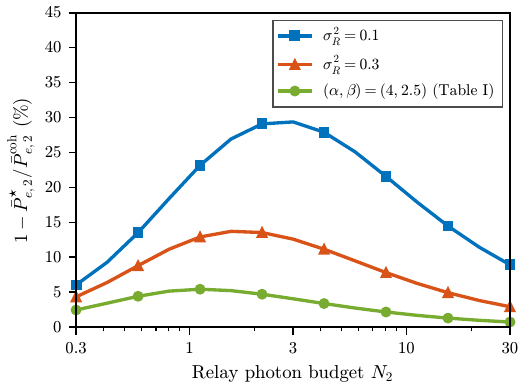}\\[-0.5mm]{\small (b)}
\caption{(a) Loss-aware optimal squeezing fraction $\zeta^\star$ versus $N_2$ for three link factors, compared with the lossless rule. (b) Relative reduction of the second-hop error achieved by optimized squeezing, $1-\bar P_{e,2}^{\star}/\bar P_{e,2}^{\mathrm{coh}}$, versus $N_2$ for $h_{\ell,2}=0.6$ and three turbulence strengths; for $\sigma_R^2\in\{0.1,0.3\}$, $(\alpha_2,\beta_2)$ follow from \eqref{eq:alpha_beta}.}
\label{fig:regime}
\end{figure}

Fig.~\ref{fig:asym} plots the second-hop error normalized by the coherent asymptotic law $\bar K(4hN_2)^{-g_2}$. Coherent and optimized DSS relaying both tend to $1$, as predicted by Theorem~\ref{thm:asymptotics}(a): optimized squeezing changes neither the diversity order $g_2$ nor the leading constant $\bar K$.  In contrast, the CP benchmark follows $(N_2+1)^{-g_2}$ and vanishes relative to the reference, as predicted by Corollary~\ref{cor:benchmark}: if squeezing were preserved, the diversity order would double.

\begin{figure}[t]
\centering
\includegraphics[width=\columnwidth]{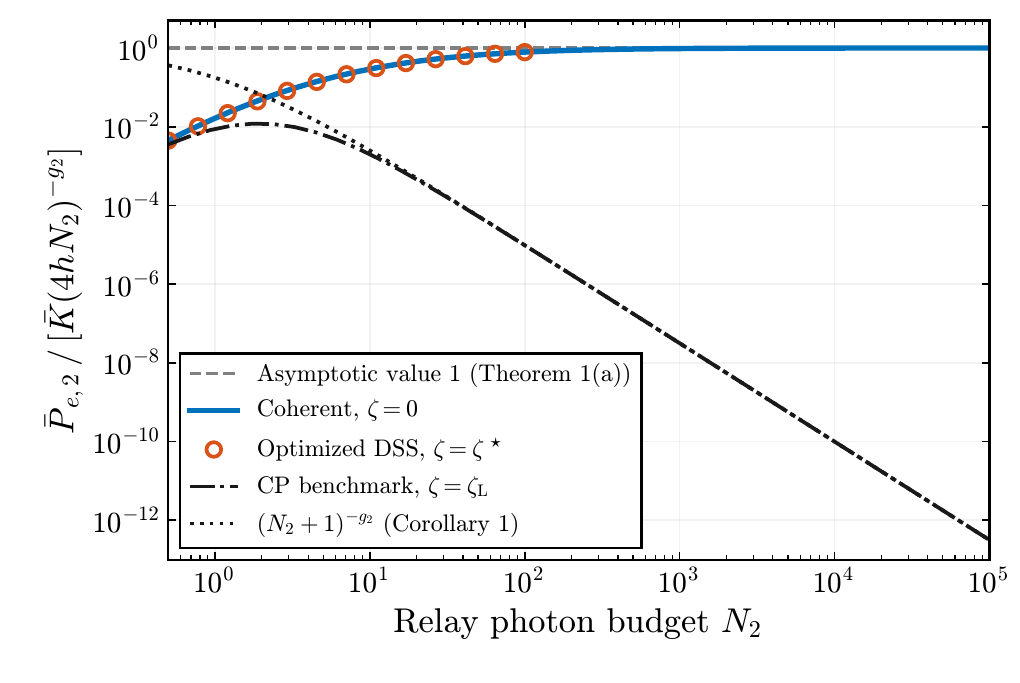}
\caption{Second-hop error normalized by the coherent asymptotic law $\bar K(4hN_2)^{-g_2}$ ($h=0.3$). Coherent and optimized DSS relaying tend to $1$ (Theorem~\ref{thm:asymptotics}(a)), whereas the CP benchmark follows $(N_2+1)^{-g_2}$ (Corollary~\ref{cor:benchmark}). DSS points are exact up to $N_2=100$; the coherent and CP curves are obtained from the series of Proposition~\ref{prop:series_laplace} up to $N_2=10^5$.}
\label{fig:asym}
\end{figure}

Fig.~\ref{fig:throughput} shows that the throughput $R_2(t)$ is unimodal with an interior maximum at a small harvesting fraction, as guaranteed by Proposition~\ref{prop:interior}: too little harvesting starves the relay, while too much leaves too little time for forwarding. The DSS curve is almost indistinguishable from the coherent one, so the throughput gain of squeezing is marginal. Fig.~\ref{fig:gainmap} shows that this gain is positive only in a thin region of small $\zeta$ and small $t$, and that the ridge of maximal gain follows $\zeta^\star(t)$. The region shrinks as $t$ grows, since $N_2(t)$ increases and $\zeta^\star\to0$ (Theorem~\ref{thm:asymptotics}(b)), and beyond it over-squeezing reduces the throughput below that of coherent relaying.

Fig.~\ref{fig:bound} plots the relative gap $R_2^{\mathrm{cp}}/R_2^{\star}-1$ versus $N_2$. The gap is nonnegative, as stated by Proposition~\ref{prop:bound}(i), and it is moderate at small budgets, where it is larger for smaller $h_\ell$ (heavier loss). It then decays with $N_2$ for all link factors, in agreement with Proposition~\ref{prop:bound}(ii): the absolute gap vanishes at the rate of the coherent second-hop error, even though the second-hop errors themselves differ by orders of magnitude (Fig.~\ref{fig:asym}). The CP benchmark is therefore a rigorous throughput upper bound that becomes tight at large budgets.

\begin{figure}[t]
\centering
\includegraphics[width=\columnwidth]{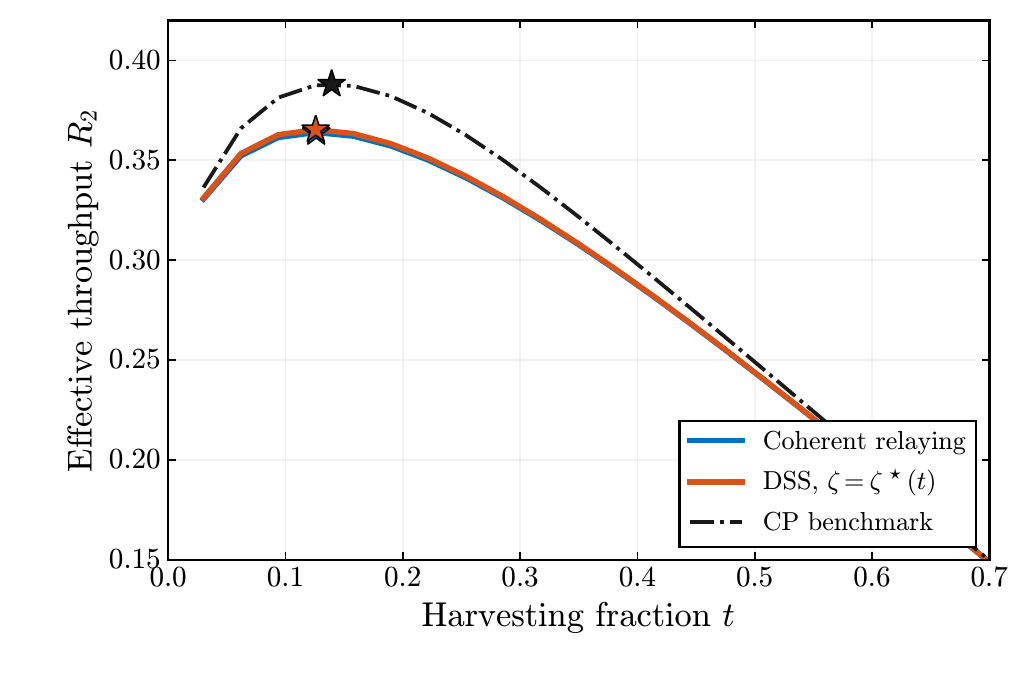}
\caption{Effective throughput $R_2(t)$ for coherent relaying, DSS relaying at $\zeta^\star(t)$ and the CP benchmark; stars mark the maximizers.}
\label{fig:throughput}
\end{figure}

\begin{figure}[t]
\centering
\includegraphics[width=\columnwidth]{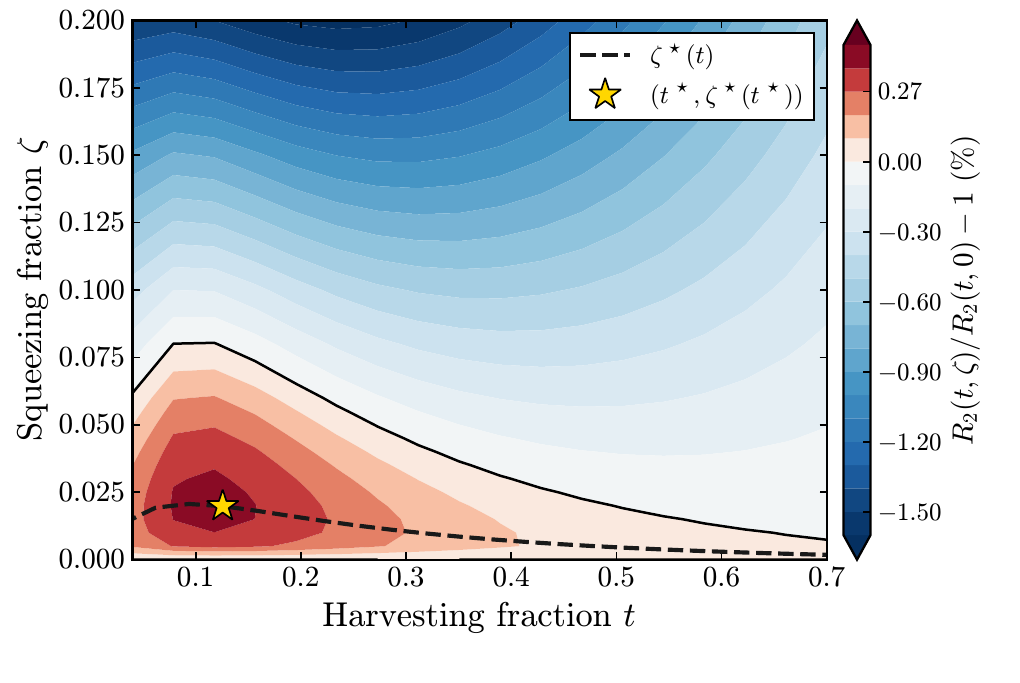}
\caption{Relative throughput gain $R_2(t,\zeta)/R_2(t,0)-1$ over $(t,\zeta)$, with the zero-gain contour (solid), the numerical ridge $\zeta^\star(t)$ (dashed) and the throughput-optimal point (star).}
\label{fig:gainmap}
\end{figure}

\begin{figure}[t]
\centering
\includegraphics[width=\columnwidth]{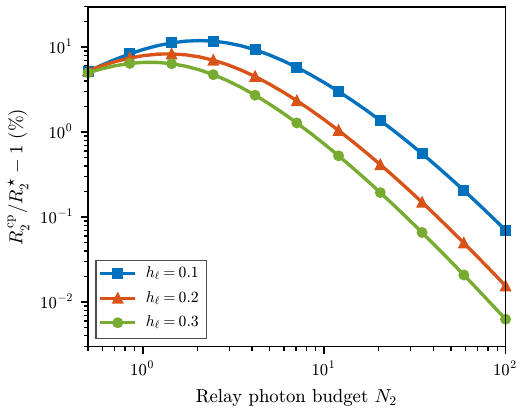}
\caption{Quality of the CP bound: relative gap $R_2^{\mathrm{cp}}/R_2^{\star}-1$ versus $N_2$ for three link factors.}
\label{fig:bound}
\end{figure}

\end{bl}

\begin{table}[!htb]
\centering
\caption{Simulation parameters.}
\label{tab:params}
\renewcommand{\arraystretch}{1.15}
\footnotesize
\begin{tabular}{@{}p{0.27\columnwidth}p{0.67\columnwidth}@{}}
\toprule
Turbulence & Gamma-Gamma, $(\alpha_i,\beta_i)=(4,2.5)$; $g_2=2.5$\\
Fig.~\ref{fig:regime}(b) & $\sigma_R^2\in\{0.1,0.3\}$, $h_{\ell,2}=0.6$\\
Link factors & $h_{\ell,1}=h_{\ell,2}=0.3$ ($F_1=F_2=0.9758$); $h_{\ell,2}\in\{0.1,0.2\}$ in Figs.~\ref{fig:regime}(a) and~\ref{fig:bound}\\
Source & BPSK coherent states, $a=16$ ($\bar P_{e,1}^{\Hel}=3.15\times10^{-3}$)\\
Relay budget & $N_2(t)=\Ct\,t/(1-t)$, $\Ct=5$\\
Throughput & $R_2(t,\zeta)=\tfrac{1-t}{2}\big(1-\bar P_e^{\DF}\big)$\\
Second-hop error & exact Helstrom error of the states of Lemma~\ref{lem:output}, averaged over \eqref{eq:eta_def}\\
\bottomrule
\end{tabular}
\end{table}
\section{Conclusion}

\begin{bl}This paper developed an analytical framework for TS-DF quantum relay-assisted FSO communication with optical EH and DSS signaling.
The proposed architecture integrates optical EH and quantum-limited detection within a unified model, where a batteryless relay harvests optical energy during a fraction $t$ of each transmission block and forwards quantum states in its own half-duplex slot of the remaining interval.
\end{bl}

\begin{bl}
Each hop was modeled as a pure-loss bosonic channel with truncated Gamma-Gamma-distributed transmissivity, so that coherent states remain pure whereas DSS states are mapped to mixed Gaussian states. For the DSS hop, we derived the Gaussian output states, a closed-form fidelity, and rigorous two-sided bounds on the Helstrom error. For coherent signaling, we obtained exact Laplace-series expressions for the average error. We proved that a loss-aware optimal squeezing allocation exists and that it differs from the known lossless rule, which under turbulence-induced loss can be worse than coherent relaying. Optimized squeezing was shown to preserve the diversity order and the leading constant of coherent relaying, so that its benefit is a finite-photon-budget effect. An interior throughput-optimal harvesting fraction exists, and the CP benchmark provides an analytical throughput upper bound whose absolute gap vanishes as the photon budget grows.

 Future research directions include excess Gaussian receiver noise or thermal background radiation, practical non-Gaussian receiver architectures, correlated fading across the two hops, and the extension of the analysis to the link-distance dependence of the end-to-end reliability.
\end{bl}

\appendices

\section{\blue{Proof of Lemma~\ref{lem:output}}}
\label{app:proof_output}

\begin{bl}
We use $\hat x=\hat a+\hat a^\dagger$, $\hat p=-j(\hat a-\hat a^\dagger)$, and the convention in which a Gaussian state has mean $\bm m=(\langle\hat x\rangle,\langle\hat p\rangle)$ and covariance $\Sigma_{kl}=\tfrac12\langle\{\Delta\hat q_k,\Delta\hat q_l\}\rangle$, so that the vacuum has $\Sigma=\mathbf I_2$ \cite{Weedbrook2012,Serafini2017}.

We first characterize the input states. For $S(r)=\exp[\frac r2(\hat a^2-\hat a^{\dagger2})]$ one has $S^\dagger\hat aS=\hat a\cosh r-\hat a^\dagger\sinh r$, hence $S^\dagger\hat xS=e^{-r}\hat x$ and $S^\dagger\hat pS=e^{r}\hat p$. For $D(\gamma)$ with $\gamma$ real, $D^\dagger\hat aD=\hat a+\gamma$, hence $D^\dagger\hat xD=\hat x+2\gamma$ and $D^\dagger\hat pD=\hat p$. The vacuum has $\langle\hat x\rangle=\langle\hat p\rangle=0$, $\langle\hat x^2\rangle=\langle\hat p^2\rangle=1$ and $\langle\{\hat x,\hat p\}\rangle=0$. Therefore, for $|\psi_b\rangle=D((-1)^b\gamma_2)S(r)|0\rangle$,
\[
\bm m_{\mathrm{in},b}=\big((-1)^b2\gamma_2,\,0\big)^{\!\top},\qquad
\Sigma_{\mathrm{in}}=\mathrm{diag}\big(e^{-2r},\,e^{2r}\big).
\]

Next, by the beamsplitter relation \eqref{eq:BS}, $\hat x_{\mathrm{out}}=\sqrt\eta\,\hat x_{\mathrm{in}}+\sqrt{1-\eta}\,\hat x_e$ and likewise for $\hat p$. The environment is in the vacuum and is uncorrelated with the input, and a beamsplitter maps Gaussian states to Gaussian states. Hence the output is Gaussian with $\bm m_{\mathrm{out}}=\sqrt\eta\,\bm m_{\mathrm{in}}$ and $\Sigma_{\mathrm{out}}=\eta\Sigma_{\mathrm{in}}+(1-\eta)\mathbf I_2$, which is \eqref{eq:Gauss_map}. Substituting the input moments above gives $\bm m_b=((-1)^b2\sqrt\eta\gamma_2,0)^\top$ and, componentwise,
\[
v_x=\eta e^{-2r}+1-\eta,\qquad v_p=\eta e^{2r}+1-\eta,
\]
which is \eqref{eq:Sigma_out}. The off-diagonal entry is zero because $\Sigma_{\mathrm{in}}$ is diagonal.

For the purity, we expand the determinant of $\Sigma=\mathrm{diag}(v_x,v_p)$ from \eqref{eq:Sigma_out}:
\begin{align*}
v_xv_p&=(1-\eta)^2+\eta^2+\eta(1-\eta)\big(e^{2r}+e^{-2r}\big)\\
&=1-2\eta(1-\eta)+2\eta(1-\eta)\cosh2r\\
&=1+2\eta(1-\eta)(\cosh2r-1)\\
&=1+4\eta(1-\eta)\sinh^2r,
\end{align*}
where we used $(1-\eta)^2+\eta^2=1-2\eta(1-\eta)$ and $\cosh2r-1=2\sinh^2r$. For a single-mode Gaussian state in this convention, $\Tr\rho^2=(\det\Sigma)^{-1/2}$ \cite{Weedbrook2012,Serafini2017}. Since $\det\Sigma\ge1$ with equality if and only if $\eta(1-\eta)\sinh^2r=0$, the state is pure if and only if $r=0$ or $\eta\in\{0,1\}$. Finally, for $r=0$ one has $\gamma_2^2=N_2$, $\Sigma=\mathbf I_2$ and $\bm m_b=((-1)^b2\sqrt{\eta N_2},0)^\top$, which is the coherent state $|(-1)^b\sqrt{\eta N_2}\rangle$. $\blacksquare$
\end{bl}

\section{\blue{Proof of Proposition~\ref{prop:fidelity}}}
\label{app:proof_fidelity}

\begin{bl}
By Lemma~\ref{lem:output}, both states share $\Sigma=\mathrm{diag}(v_x,v_p)$. Apply the single-mode squeezing $S(s)$ with $e^{4s}=v_x/v_p$ to both states. This is a fixed unitary independent of $b$, so it leaves the fidelity and the trace distance unchanged. By the squeezing relations of Appendix~\ref{app:proof_output}, it maps $\Sigma\mapsto\mathrm{diag}(e^{-2s}v_x,e^{2s}v_p)=\nu\mathbf I_2$ with $\nu=\sqrt{v_xv_p}$, and the means to $((-1)^b\delta,0)$ with $\delta=2\sqrt\eta\gamma_2e^{-s}=2\sqrt\eta\gamma_2(v_p/v_x)^{1/4}$.

Since the transformed pair is isotropic with a common covariance, its fidelity follows from the known Gaussian expression: for two single-mode Gaussian states with the same isotropic covariance $\nu\mathbf I_2$ and means differing by $\Delta\bm m$, the squared fidelity is $F=\exp(-|\Delta\bm m|^2/(4\nu))$ \cite{Scutaru1998,ParaoanuScutaru1998}. Here $|\Delta\bm m|=2\delta$, so
\[
F=\exp\!\Big(-\frac{\delta^2}{\nu}\Big)=\exp\!\Big(-\frac{4\eta\gamma_2^2\sqrt{v_p/v_x}}{\sqrt{v_xv_p}}\Big)=\exp\!\Big(-\frac{4\eta\gamma_2^2}{v_x}\Big),
\]
which proves \eqref{eq:Fclosed}. Two checks fix the constant independently of the references: at $\eta=1$, $r=0$, it gives $e^{-4\gamma_2^2}$, the overlap of $|\pm\gamma_2\rangle$, and at $\eta=1$ and general $r$ it gives $e^{-4\gamma_2^2e^{2r}}$, the overlap of the pure DSS pair. It also agrees with a direct numerical evaluation of $\|\sqrt{\rho_0}\sqrt{\rho_1}\|_1^2$.

To convert the fidelity into error bounds, let $D=\frac12\|\rho_0-\rho_1\|_1$. The Fuchs-van de Graaf inequalities \cite{FuchsVanDeGraaf1999} read $1-\sqrt F\le D\le\sqrt{1-F}$ for the squared fidelity $F$, with equality in the upper bound for pure states. For equiprobable states, $P_{e,2}^{\Hel}=\frac12(1-D)$ \cite{HelstromBook}, hence $\frac12(1-\sqrt{1-F})\le P_{e,2}^{\Hel}\le\frac12\sqrt F$, which is \eqref{eq:FvG} with $F=e^{-\Th}$. For $\eta=1$, $v_x=e^{-2r}$ and $\Th=4\gamma_2^2e^{2r}=\Theta_{\mathrm{in}}$; for $r=0$, $v_x=1$ and $\Th=4\eta\gamma_2^2=4\eta N_2$. The states are pure in both cases (Lemma~\ref{lem:output}), which gives the stated equalities.

The bounds \eqref{eq:Thbounds} follow from the expression of $v_x$ in \eqref{eq:Sigma_out}. Since $e^{-2r}\le1$, $v_x\le1$, which gives the lower bound in \eqref{eq:Thbounds}. Since $v_x\ge1-\eta$, $\Th\le4\eta\gamma_2^2/(1-\eta)$. Finally, $v_x-e^{-2r}=(1-\eta)(1-e^{-2r})\ge0$ gives $v_x\ge e^{-2r}$, hence $\Th\le4\eta\gamma_2^2e^{2r}=\eta\Theta_{\mathrm{in}}$.

For the input-level bound \eqref{eq:dataproc}, note that $\Lambda_\eta$ is completely positive and trace preserving, hence a contraction in trace norm \cite[Thm.~9.2]{NielsenChuang2010}. Therefore $\frac12\|\rho_1-\rho_0\|_1\le\frac12\||\psi_1\rangle\langle\psi_1|-|\psi_0\rangle\langle\psi_0|\|_1=\sqrt{1-|\langle\psi_0|\psi_1\rangle|^2}=\sqrt{1-e^{-\Theta_{\mathrm{in}}}}$, where the overlap of the pure DSS pair is the $\eta=1$ case of the fidelity expression above. This is \eqref{eq:dataproc}. $\blacksquare$
\end{bl}

\section{\blue{Proof of Proposition~\ref{prop:allocation}}}
\label{app:proof_allocation}

\begin{bl}
Fix $t\in(0,1)$, so that $N_2=N_2(t)$ is fixed.

We first establish the continuity of the input states in $\zeta$. The maps $\zeta\mapsto\gamma_2(\zeta)=\sqrt{(1-\zeta)N_2}$ and $\zeta\mapsto r(\zeta)=\operatorname{arsinh}\sqrt{\zeta N_2}$ are continuous on $[0,1]$. The families $\{D(\gamma)\}_{\gamma\in\mathbb R}$ and $\{S(r)\}_{r\in\mathbb R}$ are strongly continuous unitary one-parameter groups (Stone's theorem \cite{Holevo2012}). The set $\{S(r)|0\rangle:r\in[0,r_{\max}]\}$ is compact, and strong convergence is uniform on compact sets, so that
\begin{align*}
&\big\|D(\gamma')S(r')|0\rangle-D(\gamma)S(r)|0\rangle\big\|\\
&\quad\le\big\|(D(\gamma')-D(\gamma))S(r')|0\rangle\big\|+\big\|S(r')|0\rangle-S(r)|0\rangle\big\|
\end{align*}
tends to zero as $(\gamma',r')\to(\gamma,r)$. Hence $\zeta\mapsto|\psi_b(\zeta)\rangle$ is norm-continuous.

This implies the continuity of the error for fixed $\eta$: for unit vectors, $\||\psi\rangle\langle\psi|-|\psi'\rangle\langle\psi'|\|_1\le2\|\psi-\psi'\|$. The channel $\Lambda_\eta$ is a trace-norm contraction and does not depend on $\zeta$, so $\zeta\mapsto\rho_b(\eta,\zeta)$ is continuous in trace norm. Since $P_{e,2}^{\Hel}=\frac12(1-\frac12\|\rho_1-\rho_0\|_1)$ and the trace norm is continuous, $\zeta\mapsto P_{e,2}^{\Hel}(\eta;t,\zeta)$ is continuous, with values in $[0,\frac12]$.

It remains to average over the fading and to minimize. The law \eqref{eq:eta_def} of $\eta_2$ does not depend on $\zeta$, and $P_{e,2}^{\Hel}\in[0,\frac12]$ is bounded. By dominated convergence, $\zeta\mapsto\bar P_{e,2}^{\Hel}(t,\zeta)=\mathbb E[P_{e,2}^{\Hel}(\eta_2;t,\zeta)]$ is continuous on $[0,1]$. A continuous function on a compact set attains its minimum, and the set of minimizers is closed and non-empty, so it has a smallest element $\zeta^\star(t)$. $\blacksquare$
\end{bl}

\section{\blue{Proof of Proposition~\ref{prop:series_laplace}}}
\label{app:proof_prop_series}

\begin{bl}
We first expand $\vphi$. For $|z|<1$, the binomial series gives $\sqrt{1-z}=1-\sum_{k\ge1}b_kz^k$ with $b_k=(-1)^{k+1}\binom{1/2}{k}$, where $\binom{1/2}{k}=\frac{1}{k!}\prod_{j=0}^{k-1}(\frac12-j)$. For $k\ge1$ the product contains $k-1$ negative factors, so $b_k>0$. Since the partial sums of $\sum_kb_kz^k$ increase with $z\in[0,1)$ and are bounded by $1-\sqrt{1-z}\le1$, monotone convergence at $z\to1^-$ gives $\sum_kb_k=1$. Setting $z=e^{-u}$ yields $\vphi(u)=\frac12\sum_{k\ge1}b_ke^{-ku}$ for $u>0$.

Next, by the definition \eqref{eq:eta_def} of $\eta_i$, the truncated Laplace transform satisfies, for $s\ge0$,
\[
\mathbb E[e^{-s\eta_i}]=\frac1{F_i}\int_0^{1/h_{\ell,i}}e^{-sh_{\ell,i}x}f_{I_i}(x)\,dx=\frac{1}{F_i}\mathcal L_{I_i}(h_{\ell,i}s)-\tau_i(s),
\]
with $\tau_i(s)\triangleq F_i^{-1}\int_{1/h_{\ell,i}}^\infty e^{-sh_{\ell,i}x}f_{I_i}(x)\,dx\ge0$. For $x>1/h_{\ell,i}$ one has $e^{-sh_{\ell,i}x}\le e^{-s}$, so $\tau_i(s)\le F_i^{-1}e^{-s}\Pr(I_i>1/h_{\ell,i})=\frac{1-F_i}{F_i}e^{-s}$, which is \eqref{eq:Laplace_trunc}.

Averaging the series of $\vphi$ termwise, we note that all terms $b_ke^{-k\Theta\eta_i}$ are nonnegative, so by Tonelli's theorem expectation and summation may be interchanged:
\begin{align*}
\mathbb E[\vphi(\Theta\eta_i)]&=\tfrac12\sum_{k\ge1}b_k\,\mathbb E\big[e^{-k\Theta\eta_i}\big]\\
&=\frac{1}{2F_i}\sum_{k\ge1}b_k\mathcal L_{I_i}(kh_{\ell,i}\Theta)-\tau_i^{\vphi}(\Theta),
\end{align*}
with $\tau_i^{\vphi}(\Theta)=\frac12\sum_kb_k\tau_i(k\Theta)\ge0$ (both series converge, since $\mathcal L_{I_i}\le1$ and $\sum_kb_k=1$). By the bound on $\tau_i$ above, $\tau_i^{\vphi}(\Theta)\le\frac{1-F_i}{F_i}\cdot\frac12\sum_kb_ke^{-k\Theta}=\frac{1-F_i}{F_i}\vphi(\Theta)$. This is \eqref{eq:series_general}. Equations \eqref{eq:Pe1_series_final} and \eqref{eq:Pe_coh_final} follow with $\Theta=4a$ and $\Theta=4N_2$, using \eqref{eq:firsthop_average_error_final} and the fact that $\bar P_{e,2}^{\Hel}(t,0)=\mathbb E[\vphi(4N_2\eta_2)]$ (Proposition~\ref{prop:fidelity}, $r=0$). $\blacksquare$
\end{bl}

\section{\blue{Proof of Lemma~\ref{lem:laplace_asympt}}}
\label{app:proof_gamma_gamma_asympt}

\begin{bl}
Let $\nu=|\alpha-\beta|>0$ and $g=\min\{\alpha,\beta\}$, so that $(\alpha+\beta-\nu)/2=g$.

We start from the behavior of the density near the origin. Since $K_{-\nu}=K_\nu$ \cite{DLMF}, the density \eqref{eq:gamma_gamma_pdf_final} involves $K_\nu(z)$ with $z=2\sqrt{\alpha\beta x}\to0^+$, and for $\nu>0$ \cite{DLMF},
\begin{equation}
K_\nu(z)\sim2^{\nu-1}\Gamma(\nu)\,z^{-\nu},\qquad z\to0^+.
\label{eq:Bessel_small_arg}
\end{equation}
Since $z^{-\nu}=2^{-\nu}(\alpha\beta)^{-\nu/2}x^{-\nu/2}$, substituting into the density gives
\begin{align*}
f_I(x)&\sim\frac{2(\alpha\beta)^{(\alpha+\beta)/2}}{\Gamma(\alpha)\Gamma(\beta)}\cdot\frac{\Gamma(\nu)}{2}(\alpha\beta)^{-\nu/2}\,x^{\frac{\alpha+\beta-\nu}{2}-1}\\
&=c_0\,x^{g-1},\qquad x\to0^+,
\end{align*}
with $c_0=(\alpha\beta)^g\Gamma(\nu)/[\Gamma(\alpha)\Gamma(\beta)]$.

Given this behavior, we rewrite the Laplace transform with the substitution $y=sx$:
\[
s^g\mathcal L_I(s)=\int_0^\infty e^{-y}y^{g-1}\,q\!\left(\frac ys\right)dy,\qquad q(x)\triangleq\frac{f_I(x)}{x^{g-1}} .
\]
By the previous step, $q(x)\to c_0$ as $x\to0^+$, so that $q(y/s)\to c_0$ for every fixed $y>0$ as $s\to\infty$. Moreover, $q$ is bounded on $(0,\infty)$: it is continuous with a finite limit at the origin, and for large $x$ the decay $K_\nu(z)\sim\sqrt{\pi/(2z)}\,e^{-z}$ \cite[\S10.40(i)]{DLMF} makes $f_I(x)x^{1-g}$ vanish as $x\to\infty$. Hence the integrand is dominated by $\sup q\cdot e^{-y}y^{g-1}$, which is integrable. By dominated convergence,
\[
\lim_{s\to\infty}s^g\mathcal L_I(s)=c_0\int_0^\infty e^{-y}y^{g-1}dy=c_0\Gamma(g),
\]
which is \eqref{eq:laplace_asympt_final} with $c_\infty=c_0\Gamma(g)$. $\blacksquare$

\emph{Remark.} For $\alpha=\beta$, $K_0(z)\sim-\ln z$ as $z\to0^+$, which adds a logarithmic factor; this case is excluded.
\end{bl}

\section{\blue{Proof of Lemma~\ref{lem:avg_asympt}}}
\label{app:proof_avg_asympt}

\begin{bl}
Write $g=g_2$, $h=h_{\ell,2}$, $F=F_2$, $\mathcal L=\mathcal L_{I_2}$ and $\Phi(\Theta)=\mathbb E[\vphi(\Theta\eta_2)]$. By Proposition~\ref{prop:series_laplace},
\begin{align*}
\Phi(\Theta)&=\frac1{2F}\sum_{k\ge1}b_k\mathcal L(kh\Theta)-\tau_2^{\vphi}(\Theta),\\
0&\le\tau_2^{\vphi}(\Theta)\le\frac{1-F}{2F}e^{-\Theta}.
\end{align*}
Multiplying by $(h\Theta)^g$ and noting that $\Theta^ge^{-\Theta}\to0$, it suffices to show
\[
(h\Theta)^g\sum_{k\ge1}b_k\mathcal L(kh\Theta)\;\longrightarrow\;c_\infty\sum_{k\ge1}b_kk^{-g},\qquad\Theta\to\infty .
\]
Each term converges: by Lemma~\ref{lem:laplace_asympt}, for each fixed $k$, $(h\Theta)^gb_k\mathcal L(kh\Theta)=b_kk^{-g}(kh\Theta)^g\mathcal L(kh\Theta)\to b_kk^{-g}c_\infty$.

The terms are also uniformly dominated: by Lemma~\ref{lem:laplace_asympt}, there are $s_0,C>0$ with $\mathcal L(s)\le Cs^{-g}$ for $s\ge s_0$, and $\mathcal L(s)\le1\le s_0^gs^{-g}$ for $s<s_0$. Hence $\mathcal L(s)\le C's^{-g}$ for all $s>0$ with $C'=\max\{C,s_0^g\}$, and $(h\Theta)^gb_k\mathcal L(kh\Theta)\le C'b_kk^{-g}\le C'b_k$, which is summable since $\sum_kb_k=1$.

Dominated convergence for series gives the limit, and therefore $(h\Theta)^g\Phi(\Theta)\to\frac{c_\infty}{2F}\sum_kb_kk^{-g}=K/F$. The equality of the two expressions for $K$ in \eqref{eq:K_def} follows from Tonelli's theorem and $\int_0^\infty u^{g-1}e^{-ku}du=\Gamma(g)k^{-g}$, together with $c_\infty=c_0\Gamma(g)$; the integral defining $K$ is finite because $\vphi\le\frac12$ and $\vphi(u)\le\frac12e^{-u}$. $\blacksquare$
\end{bl}

\section{\blue{Proof of Theorem~\ref{thm:asymptotics}}}
\label{app:proof_asymptotics}

\begin{bl}
Write $N=N_2$, $g=g_2$, $h=h_{\ell,2}$, $A(\Theta)=\mathbb E[\vphi(\Theta\eta_2)]$, $P(N,\zeta)=\bar P_{e,2}^{\Hel}(N,\zeta)$ and $P^\star(N)=\min_\zeta P(N,\zeta)$. By Lemma~\ref{lem:avg_asympt}, $A(\Theta)\sim\Kb(h\Theta)^{-g}$ as $\Theta\to\infty$. For $\zeta=0$ the states are pure (Lemma~\ref{lem:output}), so
\begin{equation}
P(N,0)=A(4N)\sim\Kb(4hN)^{-g}.
\label{eq:step1_coh}
\end{equation}

 Fix $\epsilon\in[0,1)$ and $\delta\in(0,1)$. For every $\zeta\ge\epsilon$ and $\eta\in(0,1)$, the lower bound in \eqref{eq:FvG}, the bound $\Th\le4\eta\gamma_2^2/(1-\eta)$ of \eqref{eq:Thbounds} with $\gamma_2^2=(1-\zeta)N\le(1-\epsilon)N$, and the monotonicity of $\vphi$ give
\[
P_{e,2}^{\Hel}(\eta;N,\zeta)\ge\vphi\Big(\frac{4(1-\epsilon)N\eta}{1-\eta}\Big).
\]
For $\eta<\delta$ one has $\eta/(1-\eta)\le\eta/(1-\delta)$, so with $\Theta_{\epsilon,\delta}=4(1-\epsilon)N/(1-\delta)$ the right-hand side is at least $\vphi(\Theta_{\epsilon,\delta}\eta)$. For $\eta\ge\delta$ we simply use that $\vphi(\Theta_{\epsilon,\delta}\eta)\le\frac12e^{-\Theta_{\epsilon,\delta}\delta}$. Averaging over $\eta_2$ therefore yields, for all $\zeta\ge\epsilon$, 
\begin{equation}
P(N,\zeta)\ge A(\Theta_{\epsilon,\delta})-\tfrac12e^{-\Theta_{\epsilon,\delta}\delta}.
\label{eq:step3_L}
\end{equation}
The right-hand side does not depend on $\zeta$, and by Lemma~\ref{lem:avg_asympt} it is asymptotically equal to $\Kb\big(4h(1-\epsilon)N/(1-\delta)\big)^{-g}$, because the exponential term decays faster than any power of $N$. Dividing by \eqref{eq:step1_coh}, for every $\zeta\ge\epsilon$,
\begin{equation}
\frac{P(N,\zeta)}{P(N,0)}\ge\Big(\frac{1-\delta}{1-\epsilon}\Big)^{g}-o(1),
\label{eq:ratio_lb}
\end{equation}
where $o(1)\to0$ as $N\to\infty$ uniformly in $\zeta$.

\emph{Part (a).} Take $\epsilon=0$, so that \eqref{eq:ratio_lb} holds for all $\zeta$, in particular for a minimizer. Since $\zeta=0$ is admissible, $P^\star(N)\le P(N,0)$, hence
\[
(1-\delta)^g-o(1)\le\frac{P^\star(N)}{P(N,0)}\le1 .
\]
As $\delta\in(0,1)$ is arbitrary, $P^\star(N)/P(N,0)\to1$, and \eqref{eq:step1_coh} gives $P^\star(N)\sim P(N,0)\sim\Kb(4hN)^{-g}$.

\emph{Part (b).} Fix $\epsilon\in(0,1)$ and take $\delta=\epsilon/2$ in \eqref{eq:ratio_lb}. The constant $\big((1-\epsilon/2)/(1-\epsilon)\big)^g$ is strictly larger than $1$, so there is $N_\epsilon$ such that $P(N,\zeta)>P(N,0)$ for all $\zeta\ge\epsilon$ and $N\ge N_\epsilon$. Hence no minimizer lies in $[\epsilon,1]$, i.e., $\zeta^\star(N)<\epsilon$ for $N\ge N_\epsilon$. Since $\epsilon$ is arbitrary, $\zeta^\star(N)\to0$. $\blacksquare$
\end{bl}

\section{\blue{Proof of Proposition~\ref{prop:interior}}}
\label{app:proof_interior}

\begin{bl}
Write $P_1=\bar P_{e,1}^{\Hel}$, $B=1-2P_1$, $P_2(t)=\bar P_{e,2}^\star(t)$ and $D(t)=\frac12-P_2(t)$. Since $\vphi(u)<\frac12$ for $u>0$ and $\eta_1>0$ almost surely, $P_1<\frac12$, so $B\in(0,1]$, and $P_1$ does not depend on $t$. Using $1-P_1=\frac12+\frac B2$, \eqref{eq:R2_affine_final} becomes
\begin{equation}
R_2^\star(t)=\frac{1-t}{2}\Big(\frac12+B\,D(t)\Big),\qquad D(t)\in\Big[0,\frac12\Big],
\label{eq:R2_D}
\end{equation}
where $D\ge0$ because $\zeta=0$ gives pure states with error below $\frac12$. It therefore suffices to study $D(t)$.

\emph{Continuity.} The map $(t,\zeta)\mapsto(\gamma_2,r)$ is continuous, so by the continuity argument of Appendix~\ref{app:proof_allocation} the conditional error is jointly continuous in $(t,\zeta)$ and bounded by $\frac12$. By dominated convergence, $\bar P_{e,2}^{\Hel}(t,\zeta)$ is jointly continuous, and the minimum over the compact set $\zeta\in[0,1]$ of a jointly continuous function is continuous in $t$ (Berge's maximum theorem). Hence $P_2$, $D$ and $R_2^\star$ are continuous on $(0,1)$.

\emph{Endpoints.} Since $\frac12+BD\le1$, \eqref{eq:R2_D} gives $R_2^\star(t)\le\frac{1-t}{2}\to0$ as $t\to1^-$. As $t\to0^+$, \eqref{eq:dataproc} gives $P_2(t)\ge\vphi(\Theta_{\mathrm{in}})$ with $\Theta_{\mathrm{in}}=4(1-\zeta)N_2e^{2r}\le4N_2(\sqrt{N_2}+\sqrt{1+N_2})^2$, because $e^r=\sqrt{\zeta N_2}+\sqrt{1+\zeta N_2}$. This bound tends to $0$ with $N_2$, so $\vphi(\Theta_{\mathrm{in}})\to\frac12$, $D(t)\to0$, and $R_2^\star(t)\to\frac14$.

\emph{Excess over $\frac14$.} Let $N_2\le\frac14$, so that $4N_2\eta_2<1$. Since $1-e^{-u}\ge u/2$ on $[0,1]$, $\vphi(u)\le\frac12(1-\sqrt{u/2})$ there. As $\zeta=0$ is admissible, $P_2(t)\le\mathbb E[\vphi(4N_2\eta_2)]$, hence
\[
D(t)\ge\tfrac12\mathbb E\big[\sqrt{2N_2\eta_2}\big]=c\sqrt{N_2},\qquad c=\frac{\mathbb E[\sqrt{\eta_2}]}{\sqrt2}>0 .
\]
By \eqref{eq:R2_D}, $R_2^\star(t)>\frac14$ if and only if $BD(t)>\frac{t}{2(1-t)}$. Since $N_2=\Ct\,t/(1-t)$, we have $\frac{t}{1-t}=N_2/\Ct$, and it suffices that $Bc\sqrt{N_2}>N_2/(2\Ct)$, i.e., $\sqrt{N_2}<2\Ct Bc$. This holds for all sufficiently small $t$.

Finally, extend $R_2^\star$ to $[0,1]$ by $R_2^\star(0)=\frac14$ and $R_2^\star(1)=0$. By the continuity and the endpoint limits, the extension is continuous on the compact interval $[0,1]$, so it attains its maximum. This maximum exceeds $\frac14$ by the previous paragraph, hence it is not attained at $0$ or $1$, and it is attained at an interior point $t^\star\in(0,1)$. $\blacksquare$
\end{bl}

\section{\blue{Proof of Proposition~\ref{prop:bound}}}
\label{app:proof_bound}

\begin{bl}
We first show that the lossless partition maximizes $\Theta_{\mathrm{in}}$. Let $s=\sinh r$, so that $\gamma_2^2=N_2-s^2$ and $\Theta_{\mathrm{in}}/4=(N_2-\sinh^2r)e^{2r}$ for $r\in[0,\operatorname{arsinh}\sqrt{N_2}]$. Then
\begin{align*}
\frac{d}{dr}\Big[(N_2-\sinh^2r)e^{2r}\Big]&=2e^{2r}\big[N_2-\sinh^2r-\sinh r\cosh r\big]\\
&=2e^{2r}\Big[N_2-\frac{e^{2r}-1}{2}\Big],
\end{align*}
since $\sinh^2r+\sinh r\cosh r=\sinh r\,e^{r}=(e^{2r}-1)/2$. The bracket is decreasing in $r$ and vanishes at $e^{2r}=2N_2+1$, so $\Theta_{\mathrm{in}}$ is maximized there. At this point $\sinh^2r=(e^{2r}+e^{-2r}-2)/4$ gives $\zeta=\sinh^2r/N_2=N_2/(2N_2+1)=\zeta_{\mathrm L}$, $\gamma_2^2=N_2(N_2+1)/(2N_2+1)$ and $\Theta_{\mathrm{in}}=4N_2(N_2+1)$, as in \eqref{eq:zeta_lossless}.

For part (i), by \eqref{eq:cp_lower}, $P_{e,2}^{\mathrm{cp}}(\eta;t,\zeta)\le P_{e,2}^{\Hel}(\eta;t,\zeta)$ for all $\eta,\zeta$. Averaging preserves the inequality, so $\bar P_{e,2}^{\mathrm{cp}}(t,\zeta)\le\bar P_{e,2}^{\Hel}(t,\zeta)$ for every $\zeta$, and taking minima over $\zeta$ on both sides gives $\min_\zeta\bar P_{e,2}^{\mathrm{cp}}(t,\zeta)\le\bar P_{e,2}^\star(t)$. Since $R_2$ in \eqref{eq:R2_affine_final} is strictly decreasing in $\bar P_{e,2}$ (the coefficient $1-2\bar P_{e,1}^{\Hel}$ is positive), $R_2^{\mathrm{cp}}(t)\ge R_2^{\star}(t)$ for every $t$, and the maxima over $t$ are ordered likewise.

For part (ii), from \eqref{eq:R2_affine_final},
\[
R_2^{\mathrm{cp}}-R_2^{\star}=\frac{1-t}{2}\big(1-2\bar P_{e,1}^{\Hel}\big)\Big(\bar P_{e,2}^\star-\min_\zeta\bar P_{e,2}^{\mathrm{cp}}\Big).
\]
The CP error $\bar P_{e,2}^{\mathrm{cp}}(t,\zeta)=\mathbb E[\vphi(\eta_2\Theta_{\mathrm{in}})]$ is decreasing in $\Theta_{\mathrm{in}}$, which is maximal at $\zeta_{\mathrm L}$ by the computation above. By Lemma~\ref{lem:avg_asympt}, $\min_\zeta\bar P_{e,2}^{\mathrm{cp}}=\mathbb E[\vphi(4N_2(N_2+1)\eta_2)]\sim\Kb(4hN_2(N_2+1))^{-g_2}=o(N_2^{-g_2})$, while $\bar P_{e,2}^\star\sim\Kb(4hN_2)^{-g_2}$ by Theorem~\ref{thm:asymptotics}(a). The difference is therefore asymptotically equivalent to $\Kb(4hN_2)^{-g_2}$, which proves (ii). $\blacksquare$
\end{bl}

\balance


\begin{thebibliography}{00}

\bibitem{KhalfetQCNC2026}
N.~Khalfet and I.~Krikidis,
``Time-Switching Quantum Relay-Assisted Free-Space Optical Communication,''
in \emph{Proc. IEEE QCNC}, Kobe, Japan, 2026, pp. 529-536. 

\bibitem{b1}
L.~Gyongyosi, S.~Imre, and H.~V.~Nguyen,
``A survey on quantum channel capacities,''
\emph{IEEE Commun. Surveys Tuts.}, vol.~20, no.~2, pp.~1149--1205, Second Quart., 2018.

\bibitem{Hanzo2025}
L. Hanzo, Z. Babar, Z. Cai, D. Chandra, I. B. Djordjevic, 
B. Koczor, S. X. Ng, M. Razavi, and O. Simeone,
``Quantum information processing, sensing, and communications:
Their myths, realities, and futures,''
\emph{Proc. IEEE}, vol.~113, no.~2, pp.~222--290, Feb.~2025.



\bibitem{SchumacherWestmoreland}
B.~Schumacher and M.~D.~Westmoreland,
``Sending classical information via noisy quantum channels,''
\emph{Phys. Rev. A}, vol.~56, no.~1, pp.~131--138, Jul.~1997.

\bibitem{HelstromBook}
C. W. Helstrom,
\emph{Quantum Detection and Estimation Theory}.
New York, NY, USA: Academic Press, 1976.

\bibitem{b14}
S.~Lloyd,
``Capacity of the noisy quantum channel,''
\emph{Phys. Rev. A}, vol.~55, no.~3, pp.~1613--1622, Mar.~1997.




\bibitem{Pirandola2017}
S.~Pirandola, R.~Laurenza, C.~Ottaviani, and L.~Banchi,
``Fundamental limits of repeaterless quantum communications,''
\emph{Nat. Commun.}, vol.~8, p.~15043, Mar.~2017.

\bibitem{Braunstein2005}
S.~L.~Braunstein and P.~van~Loock,
``Quantum information with continuous variables,''
\emph{Rev. Mod. Phys.}, vol.~77, no.~2, pp.~513--577, 2005.


\bibitem{Weedbrook2012}
C. Weedbrook, S. Pirandola, R. García-Patrón, N. J. Cerf,
T. C. Ralph, J. H. Shapiro, and S. Lloyd,
``Gaussian Quantum Information,''
\emph{Rev. Mod. Phys.}, vol. 84, no. 2, pp. 621--669, 2012.



\bibitem{Chesi2018}
G. Chesi, S. Olivares, and M. G. A. Paris,
``Squeezing-Enhanced Phase-Shift-Keyed Binary Communication in Noisy Channels,''
\emph{Phys. Rev. A}, vol. 97, no. 3, p. 032315, 2018.

\bibitem{Takeoka2008}
M. Takeoka and M. Sasaki,
``Discrimination of the Binary Coherent Signal:
Gaussian-Operation Limit and Simple Non-Gaussian Near-Optimal Receivers,''
\emph{Phys. Rev. A}, vol. 78, no. 2, p. 022320, 2008.


\bibitem{FSO}
S. Pirandola,
``Limits and Security of Free-Space Quantum Communications,''
\emph{Phys. Rev. Research}, vol. 3, no. 1, p. 013279, Mar. 2021.

\bibitem{AlHabash2001}
M.~A.~Al-Habash, L.~C.~Andrews, and R.~L.~Phillips,
``Mathematical model for the irradiance probability density function of a laser beam propagating through turbulent media,''
\emph{Opt. Eng.}, vol.~40, no.~8, pp.~1554--1562, 2001.

\bibitem{Andrews2005}
L.~C.~Andrews and R.~L.~Phillips, 
\emph{Laser Beam Propagation through Random Media}, 2nd~ed. 
Bellingham, WA, USA: SPIE Press, 2005.


\bibitem{BhatnagarRelay}
M.~R.~Bhatnagar, M.~K.~Arti, and K.~Bhatnagar,
``Quantum relay-assisted free-space optical communication,''
\emph{IEEE Photon. J.}, vol.~17, no.~3, pp.~1--15, Jun.~2025.

\bibitem{Ulukus2019}
\blue{S.~Ulukus, A.~Yener, E.~Erkip, O.~Simeone, M.~Zorzi, P.~Grover, and K.~Huang, ``Energy harvesting wireless communications: A review of recent advances,'' \emph{IEEE J. Sel. Areas Commun.}, vol.~33, no.~3, pp. 360--381, Mar. 2015.}

\bibitem{Varshney2008}
L.~R.~Varshney,
``Transporting information and energy simultaneously,''
\emph{Proc. IEEE ISIT}, pp.~1612--1616, 2008.

\bibitem{Nasir2013}
A. A. Nasir, X. Zhou, S. Durrani, and R. A. Kennedy,
``Relaying Protocols for Wireless Energy Harvesting and Information Processing,''
\emph{IEEE Trans. Wireless Commun.}, vol. 12, no. 7, pp. 3622--3636, Jul. 2013.

\bibitem{KrikidisTesla}
I.~Krikidis,
``Tesla meets Helstrom: A wireless-powered quantum optical system,''
\emph{IEEE Wireless Commun. Lett.}, vol.~14, no.~12, pp.~3957--3961, Dec.~2025.

\bibitem{Paris2001}
\blue{M.~G.~A.~Paris,
``Nearly ideal binary communication in squeezed channels,''
\emph{Phys. Rev. A}, vol.~64, no.~1, p.~014304, 2001.}

\bibitem{Walsh2025}
\blue{A.~Walsh \emph{et al.},
``All-Gaussian state discrimination beyond the coherent Helstrom bound,''
\emph{arXiv:2510.20096}, 2025.}

\bibitem{SemenovVogel2009}
\blue{A.~A.~Semenov and W.~Vogel,
``Quantum light in the turbulent atmosphere,''
\emph{Phys. Rev. A}, vol.~80, no.~2, p.~021802(R), 2009.}

\bibitem{Vasylyev2016}
\blue{D.~Vasylyev, A.~A.~Semenov, and W.~Vogel,
``Atmospheric quantum channels with weak and strong turbulence,''
\emph{Phys. Rev. Lett.}, vol.~117, no.~9, p.~090501, 2016.}

\bibitem{Vasylyev2018}
\blue{D.~Vasylyev, W.~Vogel, and A.~A.~Semenov,
``Theory of atmospheric quantum channels based on the law of total probability,''
\emph{Phys. Rev. A}, vol.~97, no.~6, p.~063852, Jun.~2018.}


\bibitem{Audenaert2007}
\blue{K.~M.~R.~Audenaert, J.~Calsamiglia, R.~Mu{\~n}oz-Tapia, E.~Bagan, Ll.~Masanes, A.~Acin, and F.~Verstraete,
``Discriminating states: The quantum {C}hernoff bound,''
\emph{Phys. Rev. Lett.}, vol.~98, no.~16, p.~160501, 2007.}




\bibitem{Calsamiglia2008}
\blue{J.~Calsamiglia, R.~Mu{\~n}oz-Tapia, Ll.~Masanes, A.~Acin, and E.~Bagan,
``Quantum {C}hernoff bound as a measure of distinguishability between density matrices: Application to qubit and {G}aussian states,''
\emph{Phys. Rev. A}, vol.~77, no.~3, p.~032311, 2008.}


\bibitem{Pirandola2008}
\blue{S.~Pirandola and S.~Lloyd,
``Computable bounds for the discrimination of {G}aussian states,''
\emph{Phys. Rev. A}, vol.~78, no.~1, p.~012331, 2008.}

\bibitem{Tan2008}
S.~H.~Tan, B.~I.~Erkmen, V.~Giovannetti, S.~Guha, S.~Lloyd, 
L.~Maccone, S.~Pirandola, and J.~H.~Shapiro,
``Quantum illumination with Gaussian states,''
\emph{Phys. Rev. Lett.}, vol.~101, no.~25, p.~253601, 2008.



\bibitem{Serafini2017}
\blue{A.~Serafini,
\emph{Quantum Continuous Variables: A Primer of Theoretical Methods}.
Boca Raton, FL, USA: CRC Press, 2017.}

\bibitem{Shapiro2009}
J.~H.~Shapiro, “The quantum theory of optical communications,” 
\emph{IEEE J. Sel. Topics Quantum Electron.}, vol.~15, no.~6, 
pp.~1547--1569, Nov./Dec. 2009.

\bibitem{Glauber1963}
R.~J.~Glauber, 
“The quantum theory of optical coherence,” 
\emph{Phys. Rev.}, vol.~130, no.~6, pp.~2529--2539, Jun.~1963.

\bibitem{Yuen1980}
H.~P.~Yuen and J.~H.~Shapiro,
“Optical communication with two-photon coherent states---Part I:
Quantum-state propagation and quantum-noise reduction,”
\emph{IEEE Trans. Inf. Theory},
vol.~24, no.~6, pp.~657--668, Nov.~1978.

\bibitem{Clerckx2019}
B.~Clerckx, R.~Zhang, R.~Schober, D.~W.~K.~Ng, D.~I.~Kim, and H.~V.~Poor,
``Fundamentals of Wireless Information and Power Transfer: From RF Energy Harvester Models to Signal and System Designs,''
\textit{IEEE J. Sel. Areas Commun.}, vol.~37, no.~1, pp.~4--33, Jan.~2019.

\bibitem{Kennedy1973}
R.~S.~Kennedy,
``A near-optimum receiver for the binary coherent state quantum channel,''
\emph{MIT Res. Lab. Electron., Quarterly Progress Report}, no.~108, pp.~219--225, 1973.

\bibitem{Dolinar1973}
S.~J.~Dolinar,
``An optimum receiver for the binary coherent state quantum channel,''
\emph{MIT Res. Lab. Electron., Quart. Prog. Rep.}, no.~111, pp.~115--120, 1973.

\bibitem{Uhlmann1976}
\blue{A.~Uhlmann, ``The `transition probability' in the state space of a $*$-algebra,'' \emph{Rep. Math. Phys.}, vol.~9, no.~2, pp. 273--279, 1976.}

\bibitem{Jozsa1994}
\blue{R.~Jozsa, ``Fidelity for mixed quantum states,'' \emph{J. Mod. Opt.}, vol.~41, no.~12, pp. 2315--2323, 1994.}

\bibitem{DLMF}
\blue{F.~W.~J.~Olver \emph{et al.}, eds.,
\emph{NIST Digital Library of Mathematical Functions}, Release 1.2, 2024. [Online]. Available: https://dlmf.nist.gov/}

\bibitem{Scutaru1998}
\blue{H.~Scutaru,
``Fidelity for displaced squeezed thermal states and the oscillator semigroup,''
\emph{J. Phys. A: Math. Gen.}, vol.~31, no.~15, pp.~3659--3663, 1998.}

\bibitem{ParaoanuScutaru1998}
\blue{G.-S.~Paraoanu and H.~Scutaru,
``Bures distance between two displaced thermal states,''
\emph{Phys. Rev. A}, vol.~58, no.~2, pp.~869--874, 1998.}

\bibitem{FuchsVanDeGraaf1999}
\blue{C.~A.~Fuchs and J.~van~de~Graaf,
``Cryptographic distinguishability measures for quantum-mechanical states,''
\emph{IEEE Trans. Inf. Theory}, vol.~45, no.~4, pp.~1216--1227, May~1999.}

\bibitem{NielsenChuang2010}
\blue{M.~A.~Nielsen and I.~L.~Chuang,
\emph{Quantum Computation and Quantum Information}, 10th anniversary ed.
Cambridge, U.K.: Cambridge Univ. Press, 2010.}

\bibitem{Holevo2012}
A.~S.~Holevo, \emph{Quantum Systems, Channels, Information: A Mathematical Introduction}. 
Berlin, Germany: De Gruyter, 2012.

\end{thebibliography}
\end{document}